%% file: paper.tex
\documentclass[]{llncs}
\usepackage{pslatex}

\usepackage{paralist}
\usepackage{amssymb}
\usepackage{amsmath}
\usepackage{accents}
\usepackage{mathtools}
\usepackage{color}
\usepackage{stmaryrd}
\usepackage{leftidx}
\usepackage{hyperref}
\usepackage{float}
\usepackage{bbold}

\usepackage{tikz}
\usepackage{algorithm}
\usepackage{algorithmicx}
\usepackage{algcompatible}
\usepackage{bussproofs} 
\usepackage{nicefrac} 
\usepackage{orcidlink} 

\usepackage{lineno}

\input{commands.tex}

\usepackage{environ}
\NewEnviron{hide}{}
\newif\ifLongVersion\LongVersiontrue
\ifLongVersion
\usepackage[createShortEnv,conf={normal}]{proof-at-the-end}
\newenvironment{lemmaAtEnd}[1][]{\begin{lemmaE}}{\end{lemmaE}}
\else
\usepackage[createShortEnv,conf={end,no link to proof,restate}]{proof-at-the-end}

\fi

\renewcommand{\paragraph}[1]{\emph{#1}} 

\begin{document}

\title{Verifying Parameterized Networks Specified by Vertex-Replacement Graph Grammars}

\author{
  Radu Iosif\inst{1}\orcidlink{0000-0003-3204-3294}
  \and Arnaud Sangnier\inst{2}\orcidlink{0000-0002-6731-0340}
  \and Neven Villani\inst{1}\orcidlink{0000-0003-2726-5036}
}
\institute{
  Univ. Grenoble Alpes, CNRS, Grenoble INP, VERIMAG, 38000, France
  \texttt{\{firstname\}.\{lastname\}@univ-grenoble-alpes.fr}
  \and DIBRIS, Univ. of Genova, Italy \\
  \texttt{\{firstname\}.\{lastname\}@unige.it}
}

\maketitle

\begin{abstract}
  We consider the parametric reachability problem ($\prp{}{}{}{}$) for
  families of networks described by vertex-replacement (\vrtext) graph
  grammars, where network nodes run replicas of finite-state processes
  that communicate via binary handshaking. We show that the
  $\prp{}{}{}{}$ problem for \vrtext{} grammars can be effectively
  reduced to the $\prp{}{}{}{}$ problem for hyperedge-replacement
  (\hrtext) grammars, at the cost of introducing extra edges for
  routing messages. This transformation is motivated by the existence
  of several parametric verification techniques for families of
  networks specified by \hrtext{} grammars, or similar inductive
  formalisms. Our reduction enables applying the verification
  techniques for \hrtext{} systems to systems with dense
  architectures, such as user-specified cliques and multi-partite
  graphs.
\end{abstract}

\input{introduction}
\input{definitions}

\input{translation}

\input{conclusions}

\subsubsection*{Acknowledgements.}
This research is supported by the French National Research Agency project "Parametic Verification of Dynamic Distributed Systems" (PaVeDyS) under grant number ANR-23-CE48-0005.
We also wish to thank the anonymous reviewers for their helpful suggestions.

\subsubsection*{Disclosure of interests.}
The authors have no competing interests to declare that are relevant to the content of this article.

\appendix
\section*{Appendix}
\printProofs[lemma2]

\bibliographystyle{abbrv}
\bibliography{refs}

\end{document}

%% file: commands.tex
\newtheorem{fact}{Fact}

\newcommand{\ie}{\textit{i.e.,~}}
\newcommand{\eg}{\textit{e.g.,~}}
\newcommand{\resp}{\textit{resp.~}}

\newcommand*{\figref}[1]{\hyperref[#1]{\mbox{Figure}~\ref*{#1}}}
\newcommand*{\secref}[1]{\hyperref[#1]{\mbox{Section}~\ref*{#1}}}
\newcommand*{\lemref}[1]{\hyperref[#1]{\mbox{Lemma}~\ref*{#1}}}
\newcommand*{\thmref}[1]{\hyperref[#1]{\mbox{Theorem}~\ref*{#1}}}
\newcommand*{\propref}[1]{\hyperref[#1]{\mbox{Proposition}~\ref*{#1}}}
\newcommand*{\defref}[1]{\hyperref[#1]{\mbox{Definition}~\ref*{#1}}}
\newcommand*{\exref}[1]{\hyperref[#1]{\mbox{Example}~\ref*{#1}}}
\newcommand*{\appref}[1]{\hyperref[#1]{\mbox{Appendix}~\ref*{#1}}}
\newcommand*{\eqnref}[1]{\hyperref[#1]{\mbox{Equation}~\ref*{#1}}}
\newcommand*{\factref}[1]{\hyperref[#1]{\mbox{Fact}~\ref*{#1}}}

\newcommand{\nat}{\mathbb{N}}
\newcommand{\ints}{\mathbb{Z}}
\newcommand{\interv}[2]{{[{#1},{#2}]}}
\newcommand{\isdef}{\stackrel{\scriptscriptstyle{\mathsf{def}}}{=}}
\newcommand{\iffdef}{\stackrel{\scriptscriptstyle{\mathsf{def}}}{\iff}}
\newcommand{\set}[1]{{\{ {#1} \}}}

\newcommand{\pow}[1]{{\mathcal{P}({#1})}}
\newcommand{\cardof}[1]{|\!|{#1}|\!|}
\newcommand{\arrow}[2]{\xrightarrow{{\scriptscriptstyle #1}}_{{\scriptscriptstyle #2}}}
\renewcommand{\vec}[1]{\mathbf{#1}}
\newcommand{\proj}[2]{{#1}\!\downharpoonleft_{#2}}
\newcommand{\fnid}{\mathrm{id}}
\newcommand{\sem}[1]{[\![{#1}]\!]}
\newcommand{\mso}{$\mathsf{MSO}$}

\newcommand{\anet}{\mathsf{N}}
\newcommand{\amarkednet}{\mathcal{N}}
\newcommand{\places}{\mathsf{Q}}
\newcommand{\placeof}[1]{\places_{#1}}
\newcommand{\trans}{\mathsf{T}}
\newcommand{\transof}[1]{\trans_{#1}}
\newcommand{\weight}{\mathsf{W}}
\newcommand{\weightof}[1]{\weight_{#1}}
\newcommand{\initmark}{\mathrm{init}}
\newcommand{\initmarkof}[1]{\initmark_{#1}}
\newcommand{\pre}[1]{\text{$\leftidx{^\bullet}{\text{${#1}$}}$}}
\newcommand{\post}[1]{\text{${#1}$}^\bullet}
\newcommand{\prepost}[1]{\leftidx{^\bullet}{\!\text{${#1}$}}^\bullet}
\newcommand{\amark}{\mathsf{m}}

\newcommand{\fire}[1]{\stackrel{\scriptscriptstyle{#1}}{\leadsto}}

\newcommand{\ptype}{\mathsf{p}}
\newcommand{\ptypeof}[1]{\mathrm{ptype}_{#1}}
\newcommand{\portof}[1]{\mathrm{port}_{#1}}
\newcommand{\procof}[1]{\mathrm{proc}_{#1}}
\newcommand{\ptypes}{\mathbb{P}}
\newcommand{\obstransof}[1]{\trans^\mathit{obs}_{#1}}
\newcommand{\inttransof}[1]{\trans^\mathit{int}_{#1}}
\newcommand{\graph}{\mathsf{G}}
\newcommand{\verts}{\mathsf{V}}
\newcommand{\vertof}[1]{\verts_{#1}}
\newcommand{\edges}{\mathsf{E}}
\newcommand{\edgeof}[1]{\edges_{#1}}
\newcommand{\labels}{\lambda}
\newcommand{\labof}[1]{\labels_{#1}}
\newcommand{\asys}{\mathsf{S}}
\newcommand{\systems}[1]{\mathcal{S}({#1})}
\newcommand{\valpha}{\Lambda}
\newcommand{\vlab}{\lambda}
\newcommand{\palpha}{\Pi}
\newcommand{\expalpha}{\mathcal{E}}
\newcommand{\plab}{\pi}
\newcommand{\plabs}{\mathbf{p}}

\newcommand{\ealpha}{\Delta}
\newcommand{\elab}{\delta}
\newcommand{\abeh}{\beta}
\newcommand{\behof}[1]{\abeh({#1})}
\newcommand{\graphs}[2]{\mathcal{G}\ifthenelse{\equal{#1}{}\AND\equal{#2}{}}{}{({#1},{#2})}}
\newcommand{\hrgraphs}[2]{\mathcal{G}^{\scriptscriptstyle{\hr}}\ifthenelse{\equal{#1}{}\AND\equal{#2}{}}{}{({#1},{#2})}}
\newcommand{\epsgraphs}[2]{\mathcal{G}^{\epsilon}\ifthenelse{\equal{#1}{}\AND\equal{#2}{}}{}{({#1},{#2})}}
\newcommand{\expof}[1]{\mathrm{exp}({#1})}

\newcommand{\hr}{\mathsf{HR}}

\newcommand{\hrtext}{\textsf{HR}}
\newcommand{\vr}{\mathsf{VR}}
\newcommand{\vrtext}{\textsf{VR}}

\newcommand{\hrpop}[2]{\parallel_{\ifx#1\empty\else{\scriptscriptstyle{{#1},{#2}}}\fi}}
\newcommand{\bighrpop}[1]{\bigparallel_{#1}}
\newcommand{\edge}[3]{\overrightarrow{#1}_{{#2},{#3}}}

\newcommand{\vrpop}[2]{\oplus_{\ifx#1\empty\else{\scriptscriptstyle{{#1},{#2}}}\fi}}
\newcommand{\bigvrpop}[1]{\bigoplus_{\scriptscriptstyle{{#1}}}}
\newcommand{\svertex}[1]{\bullet{#1}}
\newcommand{\addedge}[3]{\mathsf{add}_{{#1}\ifx#2\empty\else{,{#2},{#3}}\fi}}
\newcommand{\relab}[1]{\mathsf{relab}_{#1}}

\newcommand{\grammar}{\Gamma}
\newcommand{\rules}{\Pi}
\newcommand{\nonterm}{\Xi}
\newcommand{\step}[1]{\Rightarrow_{#1}}
\newcommand{\aop}{\mathsf{op}}
\newcommand{\alangof}[2]{\mathcal{L}_{#1}({#2})}
\newcommand{\langu}{\mathcal{L}}

\newcommand{\val}[2]{\mathrm{val}^{\scriptscriptstyle{#1}}\ifthenelse{\equal{#2}{}}{}{({#2})}}
\newcommand{\halve}[1]{{#1}_{\scriptscriptstyle{\nicefrac{1}{2}}}}
\newcommand{\expand}{\mathcal{H}} 

\newcommand{\ctlstarnext}{$\mathsf{CTL}^*\setminus\bigcirc$} 
\newcommand{\atoms}{\mathbb{A}}
\newcommand{\atomsof}[2]{\mathit{Atoms}^{#2}({#1})}
\newcommand{\pathsof}[1]{\mathit{Paths}({#1})}
\newcommand{\vars}{\mathbb{X}}
\newcommand{\varlab}{\ell}
\newcommand{\syslab}{\mathsf{L}}
\newcommand{\stutter}{\trianglelefteq}
\newcommand{\sameap}{\approx}

\newcommand{\gtidle}{\mathsf{idle}}
\newcommand{\gtactive}{\mathsf{active}}
\newcommand{\gtwait}{\mathsf{wait}}
\newcommand{\gtreply}{\mathsf{reply}}
\newcommand{\gttry}[1]{{\mathsf{try}}_{#1}}
\newcommand{\halfplace}[1]{{#1}_{\scriptscriptstyle{\nicefrac{1}{2}}}}
\newcommand{\gtcommit}[1]{\mathsf{commit}_{#1}}
\newcommand{\gtrecv}{{\mathsf{recv}}}
\newcommand{\gtfwd}{{\mathsf{send}}} 
\newcommand{\gtack}{\mathsf{ack}}
\newcommand{\gtreset}{\mathsf{reset}}

\renewcommand{\partial}{\rightharpoonup}

\newcommand{\pmcp}[4]{\mathsf{PMCP}_{#4}\ifthenelse{\equal{#1}{}\AND\equal{#2}{}\AND\equal{#3}{}}{}{({#1},{#2},{#3})}}
\newcommand{\prp}[4]{\mathsf{PRP}_{#4}\ifthenelse{\equal{#1}{}\AND\equal{#2}{}\AND\equal{#3}{}}{}{({#1},{#2},{#3})}}

\usetikzlibrary{arrows,shapes,decorations,automata,backgrounds,fit,matrix}
\usetikzlibrary{calc,decorations.pathreplacing}
\usetikzlibrary{petri}

\tikzstyle{petri-p}=[circle,thick,draw=black,inner sep=0pt,minimum size=5mm]
\tikzstyle{petri-thor}=[rectangle,draw=black,thick,inner sep=0pt,minimum width=5mm,minimum height=1mm]
\tikzstyle{petri-tver}=[rectangle,draw=black,thick,inner sep=0pt,minimum width=1mm,minimum height=5mm]
\tikzstyle{petri-tpos}=[rectangle,draw=black,thick,inner sep=0pt,minimum width=1mm,minimum height=5mm,rotate=-45]
\tikzstyle{petri-tneg}=[rectangle,draw=black,thick,inner sep=0pt,minimum width=1mm,minimum height=5mm,rotate=45]
\tikzstyle{petri-tok}=[circle,inner sep=0pt,minimum size=4pt, color=black,fill=black]
\tikzstyle{petri-small-tok}=[circle,inner sep=0pt,minimum size=3pt, color=black,fill=black]        
\tikzstyle{gnode}=[circle,inner sep=0pt,minimum size=4pt, color=black,fill=black]

%% file: introduction.tex
\section{Introduction}

As pointed out in recent literature, ``\emph{network verification is a
  necessary part of deploying modern hyperscale
  datacenters}''~\cite{DBLP:conf/sigcomm/JayaramanBPABBF19}. In this
context, verification considers global properties of routing control,
such as access between point A and point B. Very often, the intended
functionality of a network can be derived from its architecture,
typically known at early stages of network design.

The huge scale of present-day datacenters, of $\sim\!\!10^4$ routers for a
regional hub and continuously growing, requires \emph{parametric}
verification techniques, that work no matter how many servers in a
rack, switches in a cluster, clusters in a layer, etc.
Despite decades of theoretical research and an impressive body
of results (see~\cite{BloemJacobsKhalimovKonnovRubinVeithWidder15} for
a nice survey), these techniques consider mostly hard-coded
architectures, either dense (\eg cliques~\cite{GermanSistla92}) or
sparse (\eg rings~\cite{ClarkeGrumbergBrowne86}) families of graphs
with a common topological pattern.

The problem of verifying networks with user-specified architecture has
gained recent momentum, with the development of parametric
verification techniques for systems described using logic
\cite{DBLP:journals/dc/AminofKRSV18} or graph grammars
\cite{LeMetayer}. In principle, logic-based graph specification
languages are good at describing global properties, such as planarity,
$k$-colorability or reachability, which make them more suitable for
the specification of correctness properties, than network design.

On the other hand, graph grammars are appealing for their constructive
aspect, \ie defining how large graphs are built from smaller
subgraphs. Moreover, this recursive way of describing sets of graphs
is common among programmers and software engineers, who are familiar
with inductively defined datastructures (lists, trees, etc.)
used in both imperative and functional programming.
Graph grammars are at the core of the theory of formal languages
(see~\cite{courcelle_engelfriet_2012} for a comprehensive
survey). Based on the underlying set of operations (\ie graph
algebra), we distinguish between \emph{vertex-replacement} (\vrtext)
and \emph{hyperedge-replacement} (\hrtext) grammars. In principle,
\hrtext\ graph grammars can specify families of graphs having bounded
tree-width, such as chains, rings, stars, trees (of unbounded rank)
and beyond, \eg overlaid structures such as trees or stars with
certain nodes linked in a list. Since cliques and grids are families
of unbounded tree-width, neither can be specified using \hrtext\ graph
grammars. On the other hand, \emph{vertex-replacement} (\vrtext)
grammars are strictly more expressive~\cite{COURCELLE1995275}. For
instance, \vrtext{} grammars can describe the sets $\set{K_n}_{n \geq
  1}$ of cliques and $\set{K_{n,m}}_{n,m \geq 1}$ of complete
bipartite graphs, among others.

A concrete motivation for using \vrtext{} grammars to design
datacenter networks can be found in the informal description of the
Azure datacenter architecture
~\cite{10.1145/1594977.1592576,DBLP:conf/sigcomm/JayaramanBPABBF19},
deployed in the Microsoft hyperscale cloud network
(\figref{fig:azure}). The \emph{top-of-rack} (ToR) switches connect
servers hosted in a rack. Several ToR switches are connected together
by a number of \emph{leaf} switches. Leaf switches are in turn
connected together by a set of \emph{spine} switches, that connect the
datacenter to the Azure regional network. These networks cannot be
described as cliques, because of their layered structure. However, the
links between each layer (\ie ToR, leaf, spine, regional spine) and its
adjacent layers form a complete bipartite graph, thus lying outside
the scope of \hrtext{} grammars, which cannot define complete
bipartite graphs of unbounded sizes\footnote{Each set
$\set{K_{n,m}}_{n\in N,m\in M}$ for $N,M$ infinite subsets of $\nat$,
has unbounded tree-width, whereas \hrtext{}-grammars define bounded
tree-width sets.}.

\begin{figure}[t!]
  \vspace*{-\baselineskip}
  \centerline{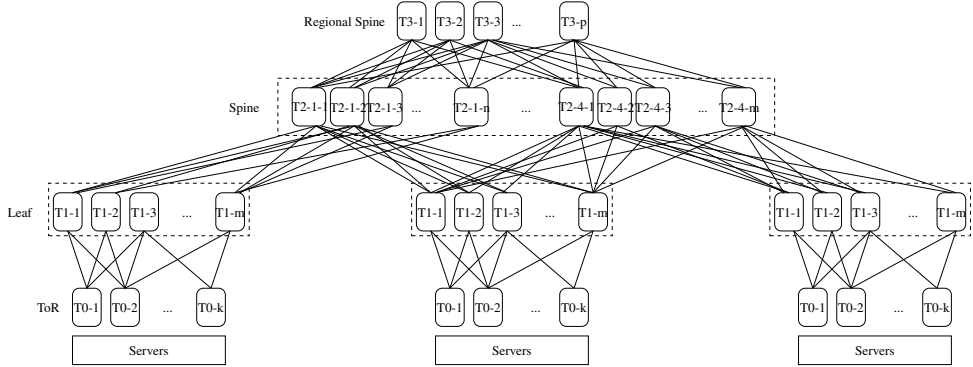}
  \vspace*{-.5\baselineskip}
  \caption{Azure Datacenter Switching Topology}
  \label{fig:azure}
  \vspace*{-\baselineskip}
\end{figure}

The precise relation between the expressivity of \vrtext{} and
\hrtext{} grammars has been elegantly coined in a result by Courcelle,
stating that the sets defined by \hrtext{} grammars are exactly the
sparse sets defined by \vrtext{}
grammars~\cite{COURCELLE1995275}. Moreover, each \vrtext{} grammar can
be transformed into an \hrtext{} grammar producing a set of graphs
with $\epsilon$-edges, such that the elimination of these edges
from the \hrtext{} language yields the original \vrtext{}
language. For instance, the elimination of $\epsilon$-edges from the
graph from \figref{fig:K43} (b) yields the complete bipartite graph
$K_{4,3}$ from \figref{fig:K43} (a).

\input{figure-term-example}~

\paragraph{Contribution}
We leverage the idea of transforming \vrtext{} grammars into
\hrtext{} grammars modulo elimination of $\epsilon$-edges to define a
reduction from the parametric reachability problem ($\prp{}{}{}{}$)
for systems described by \vrtext{} grammars and finite-state local
behaviors, with correctness properties given by unrestricted
arithmetic formul{\ae} over a set of variables, tracking how many
processes are in a given local state, to the $\prp{}{}{}{}$ problem
for systems described by \hrtext{} grammars, having the same
correctness property. This is motivated by several recent developments
in the verification of parametric systems with \hrtext{}-style
architecture descriptions. In
particular,~\cite{DBLP:journals/tcs/BozgaIS23} reports on an inference
of structural invariants (\ie over-approximations of the reachable
set, that can be derived from the structure of the network) for
architectures described using a variant of Separation Logic with
inductive definitions (same expressivity as \hrtext{} grammars)
Moreover,~\cite{arXiv2025} gives an abstraction technique
that folds similar network nodes into a finite Petri net, whose
coverability problem over-approximates the original parametric
coverability problem. In this paper, we give a more general reduction
method, that preserves all safety properties specified by arithmetic
assertions, rather than the particular case of coverability.

\paragraph{Related Work}
The literature on parametric verification is too large to be recalled
here. We point to \cite{BloemJacobsKhalimovKonnovRubinVeithWidder15}
for a survey on the (un)decidability results, typically concerning
clique or ring architectures. The specification of network
architectures using inductive definitions can be traced back to the
work on \emph{network
  grammars}~\cite{ShtadlerGrumberg89,LeMetayer,Hirsch}, that use
inductive rules to describe systems with linear (pipeline, token-ring)
architectures. These initial works target mainly safety properties, by
inferring network
invariants~\cite{WolperLovinfosse89,LesensHalbwachsRaymond97}. Pipeline
and clique architectures are also considered by methods based on the
theory of well-structured transition
systems~\cite{DBLP:conf/lics/AbdullaCJT96}, such as monotonic
abstraction~\cite{abdulla-monotonic-foundcs09,abdulla-approximated.fmsd09}.

A prominent exception, that considers the verification of clique-width
bounded architectures specified using Monadic Second Order Logic
(\mso) against \ctlstarnext{} properties
is~\cite{DBLP:journals/dc/AminofKRSV18}. Their result uses a reduction
from the parametric model checking problem ($\pmcp{}{}{}{}$) to the decidability of
\mso{} over graphs of bounded
clique-width~\cite{courcelle_engelfriet_2012}. Our result is
orthogonal, because we consider \vrtext{} grammars that strictly
subsume the \mso-definable sets of bounded clique-width. It is,
however, an interesting question if our reduction can be used to
establish a more general decidability result. A promising lead would
be to use the decidability of $\pmcp{}{}{}{}$ for \hrtext{} systems
with local behaviors that mimick the passing of a pebble from one node
to another~\cite{arXiv2025}.

\paragraph{Notation}
We denote by $\ints$ ($\nat$) the set of (positive) integers. For $i,
j \in \nat$, we denote by $\interv{i}{j}$ the set $\set{i,\ldots,j}$,
considered empty if $i > j$. The cardinality of a finite set $A$ is
written $\cardof{A}$. A singleton $\set{a}$ is simply denoted $a$.
The union of two disjoint sets $A$ and $B$ is denoted as $A \uplus B$.
The Cartesian product of two sets $A$ and $B$ is denoted $A \times B$.
As usual, we define $A^*$ the set of (possibly empty) finite sequences
of elements from $A$, and $A^\infty$ the set of finite or infinite
sequences over $A$. Finally $\pow{A}$ is the powerset of $A$.

%% file: azure.pdf_t
\begin{picture}(0,0)%
\includegraphics{azure.pdf}%
\end{picture}%
\setlength{\unitlength}{1579sp}%
\begingroup\makeatletter\ifx\SetFigFont\undefined%
\gdef\SetFigFont#1#2#3#4#5{%
  \reset@font\fontsize{#1}{#2pt}%
  \fontfamily{#3}\fontseries{#4}\fontshape{#5}%
  \selectfont}%
\fi\endgroup%
\begin{picture}(15160,5724)(153,-8473)
\put(11671,-4486){\makebox(0,0)[b]{\smash{{\SetFigFont{5}{6.0}{\rmdefault}{\mddefault}{\updefault}{\color[rgb]{0,0,0}T2-4-m}%
}}}}
\put(6526,-3136){\makebox(0,0)[b]{\smash{{\SetFigFont{5}{6.0}{\rmdefault}{\mddefault}{\updefault}{\color[rgb]{0,0,0}T3-1}%
}}}}
\put(7126,-3136){\makebox(0,0)[b]{\smash{{\SetFigFont{5}{6.0}{\rmdefault}{\mddefault}{\updefault}{\color[rgb]{0,0,0}T3-2}%
}}}}
\put(7726,-3136){\makebox(0,0)[b]{\smash{{\SetFigFont{5}{6.0}{\rmdefault}{\mddefault}{\updefault}{\color[rgb]{0,0,0}T3-3}%
}}}}
\put(9076,-3136){\makebox(0,0)[b]{\smash{{\SetFigFont{5}{6.0}{\rmdefault}{\mddefault}{\updefault}{\color[rgb]{0,0,0}T3-p}%
}}}}
\put(8176,-3136){\makebox(0,0)[b]{\smash{{\SetFigFont{5}{6.0}{\rmdefault}{\mddefault}{\updefault}{\color[rgb]{0,0,0}...}%
}}}}
\put(5476,-3136){\makebox(0,0)[b]{\smash{{\SetFigFont{5}{6.0}{\rmdefault}{\mddefault}{\updefault}{\color[rgb]{0,0,0}Regional Spine}%
}}}}
\put(826,-7636){\makebox(0,0)[b]{\smash{{\SetFigFont{5}{6.0}{\rmdefault}{\mddefault}{\updefault}{\color[rgb]{0,0,0}ToR}%
}}}}
\put(2401,-8311){\makebox(0,0)[b]{\smash{{\SetFigFont{5}{6.0}{\rmdefault}{\mddefault}{\updefault}{\color[rgb]{0,0,0}Servers}%
}}}}
\put(2026,-7636){\makebox(0,0)[b]{\smash{{\SetFigFont{5}{6.0}{\rmdefault}{\mddefault}{\updefault}{\color[rgb]{0,0,0}T0-2}%
}}}}
\put(3376,-7636){\makebox(0,0)[b]{\smash{{\SetFigFont{5}{6.0}{\rmdefault}{\mddefault}{\updefault}{\color[rgb]{0,0,0}T0-k}%
}}}}
\put(8701,-6136){\makebox(0,0)[b]{\smash{{\SetFigFont{5}{6.0}{\rmdefault}{\mddefault}{\updefault}{\color[rgb]{0,0,0}...}%
}}}}
\put(14401,-6136){\makebox(0,0)[b]{\smash{{\SetFigFont{5}{6.0}{\rmdefault}{\mddefault}{\updefault}{\color[rgb]{0,0,0}...}%
}}}}
\put(376,-6136){\makebox(0,0)[b]{\smash{{\SetFigFont{5}{6.0}{\rmdefault}{\mddefault}{\updefault}{\color[rgb]{0,0,0}Leaf}%
}}}}
\put(1426,-7636){\makebox(0,0)[b]{\smash{{\SetFigFont{5}{6.0}{\rmdefault}{\mddefault}{\updefault}{\color[rgb]{0,0,0}T0-1}%
}}}}
\put(8101,-8311){\makebox(0,0)[b]{\smash{{\SetFigFont{5}{6.0}{\rmdefault}{\mddefault}{\updefault}{\color[rgb]{0,0,0}Servers}%
}}}}
\put(7126,-7636){\makebox(0,0)[b]{\smash{{\SetFigFont{5}{6.0}{\rmdefault}{\mddefault}{\updefault}{\color[rgb]{0,0,0}T0-1}%
}}}}
\put(7726,-7636){\makebox(0,0)[b]{\smash{{\SetFigFont{5}{6.0}{\rmdefault}{\mddefault}{\updefault}{\color[rgb]{0,0,0}T0-2}%
}}}}
\put(9076,-7636){\makebox(0,0)[b]{\smash{{\SetFigFont{5}{6.0}{\rmdefault}{\mddefault}{\updefault}{\color[rgb]{0,0,0}T0-k}%
}}}}
\put(13801,-8311){\makebox(0,0)[b]{\smash{{\SetFigFont{5}{6.0}{\rmdefault}{\mddefault}{\updefault}{\color[rgb]{0,0,0}Servers}%
}}}}
\put(12826,-7636){\makebox(0,0)[b]{\smash{{\SetFigFont{5}{6.0}{\rmdefault}{\mddefault}{\updefault}{\color[rgb]{0,0,0}T0-1}%
}}}}
\put(13426,-7636){\makebox(0,0)[b]{\smash{{\SetFigFont{5}{6.0}{\rmdefault}{\mddefault}{\updefault}{\color[rgb]{0,0,0}T0-2}%
}}}}
\put(14776,-7636){\makebox(0,0)[b]{\smash{{\SetFigFont{5}{6.0}{\rmdefault}{\mddefault}{\updefault}{\color[rgb]{0,0,0}T0-k}%
}}}}
\put(3001,-6136){\makebox(0,0)[b]{\smash{{\SetFigFont{5}{6.0}{\rmdefault}{\mddefault}{\updefault}{\color[rgb]{0,0,0}...}%
}}}}
\put(1126,-6136){\makebox(0,0)[b]{\smash{{\SetFigFont{5}{6.0}{\rmdefault}{\mddefault}{\updefault}{\color[rgb]{0,0,0}T1-1}%
}}}}
\put(1726,-6136){\makebox(0,0)[b]{\smash{{\SetFigFont{5}{6.0}{\rmdefault}{\mddefault}{\updefault}{\color[rgb]{0,0,0}T1-2}%
}}}}
\put(2326,-6136){\makebox(0,0)[b]{\smash{{\SetFigFont{5}{6.0}{\rmdefault}{\mddefault}{\updefault}{\color[rgb]{0,0,0}T1-3}%
}}}}
\put(3676,-6136){\makebox(0,0)[b]{\smash{{\SetFigFont{5}{6.0}{\rmdefault}{\mddefault}{\updefault}{\color[rgb]{0,0,0}T1-m}%
}}}}
\put(6826,-6136){\makebox(0,0)[b]{\smash{{\SetFigFont{5}{6.0}{\rmdefault}{\mddefault}{\updefault}{\color[rgb]{0,0,0}T1-1}%
}}}}
\put(7426,-6136){\makebox(0,0)[b]{\smash{{\SetFigFont{5}{6.0}{\rmdefault}{\mddefault}{\updefault}{\color[rgb]{0,0,0}T1-2}%
}}}}
\put(8026,-6136){\makebox(0,0)[b]{\smash{{\SetFigFont{5}{6.0}{\rmdefault}{\mddefault}{\updefault}{\color[rgb]{0,0,0}T1-3}%
}}}}
\put(9376,-6136){\makebox(0,0)[b]{\smash{{\SetFigFont{5}{6.0}{\rmdefault}{\mddefault}{\updefault}{\color[rgb]{0,0,0}T1-m}%
}}}}
\put(12451,-6136){\makebox(0,0)[b]{\smash{{\SetFigFont{5}{6.0}{\rmdefault}{\mddefault}{\updefault}{\color[rgb]{0,0,0}T1-1}%
}}}}
\put(13051,-6136){\makebox(0,0)[b]{\smash{{\SetFigFont{5}{6.0}{\rmdefault}{\mddefault}{\updefault}{\color[rgb]{0,0,0}T1-2}%
}}}}
\put(13651,-6136){\makebox(0,0)[b]{\smash{{\SetFigFont{5}{6.0}{\rmdefault}{\mddefault}{\updefault}{\color[rgb]{0,0,0}T1-3}%
}}}}
\put(15001,-6136){\makebox(0,0)[b]{\smash{{\SetFigFont{5}{6.0}{\rmdefault}{\mddefault}{\updefault}{\color[rgb]{0,0,0}T1-m}%
}}}}
\put(2701,-7636){\makebox(0,0)[b]{\smash{{\SetFigFont{5}{6.0}{\rmdefault}{\mddefault}{\updefault}{\color[rgb]{0,0,0}...}%
}}}}
\put(8476,-7636){\makebox(0,0)[b]{\smash{{\SetFigFont{5}{6.0}{\rmdefault}{\mddefault}{\updefault}{\color[rgb]{0,0,0}...}%
}}}}
\put(14176,-7636){\makebox(0,0)[b]{\smash{{\SetFigFont{5}{6.0}{\rmdefault}{\mddefault}{\updefault}{\color[rgb]{0,0,0}...}%
}}}}
\put(8326,-4486){\makebox(0,0)[b]{\smash{{\SetFigFont{5}{6.0}{\rmdefault}{\mddefault}{\updefault}{\color[rgb]{0,0,0}...}%
}}}}
\put(11101,-4486){\makebox(0,0)[b]{\smash{{\SetFigFont{5}{6.0}{\rmdefault}{\mddefault}{\updefault}{\color[rgb]{0,0,0}...}%
}}}}
\put(3901,-4486){\makebox(0,0)[b]{\smash{{\SetFigFont{5}{6.0}{\rmdefault}{\mddefault}{\updefault}{\color[rgb]{0,0,0}Spine}%
}}}}
\put(6601,-4486){\makebox(0,0)[b]{\smash{{\SetFigFont{5}{6.0}{\rmdefault}{\mddefault}{\updefault}{\color[rgb]{0,0,0}...}%
}}}}
\put(4936,-4471){\makebox(0,0)[b]{\smash{{\SetFigFont{5}{6.0}{\rmdefault}{\mddefault}{\updefault}{\color[rgb]{0,0,0}T2-1-1}%
}}}}
\put(5536,-4486){\makebox(0,0)[b]{\smash{{\SetFigFont{5}{6.0}{\rmdefault}{\mddefault}{\updefault}{\color[rgb]{0,0,0}T2-1-2}%
}}}}
\put(6106,-4486){\makebox(0,0)[b]{\smash{{\SetFigFont{5}{6.0}{\rmdefault}{\mddefault}{\updefault}{\color[rgb]{0,0,0}T2-1-3}%
}}}}
\put(7471,-4486){\makebox(0,0)[b]{\smash{{\SetFigFont{5}{6.0}{\rmdefault}{\mddefault}{\updefault}{\color[rgb]{0,0,0}T2-1-n}%
}}}}
\put(9121,-4471){\makebox(0,0)[b]{\smash{{\SetFigFont{5}{6.0}{\rmdefault}{\mddefault}{\updefault}{\color[rgb]{0,0,0}T2-4-1}%
}}}}
\put(9736,-4486){\makebox(0,0)[b]{\smash{{\SetFigFont{5}{6.0}{\rmdefault}{\mddefault}{\updefault}{\color[rgb]{0,0,0}T2-4-2}%
}}}}
\put(10321,-4486){\makebox(0,0)[b]{\smash{{\SetFigFont{5}{6.0}{\rmdefault}{\mddefault}{\updefault}{\color[rgb]{0,0,0}T2-4-3}%
}}}}
\end{picture}%

%% file: figure-term-example.tex
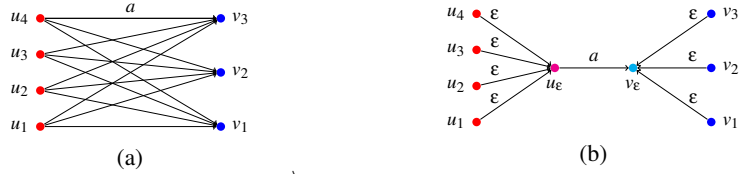
\begin{figure}[t!]
  \vspace*{-\baselineskip}
  \begin{minipage}{0.5\textwidth}
    \begin{center}
    \scalebox{0.8}{\begin{tikzpicture}
      \tikzset{
        point/.style={circle,inner sep=0pt,minimum size=4pt},
        thickpoint/.style={circle,inner sep=0pt,minimum size=15pt},
        arr/.style={->,line width=0.1pt}
      }
      \foreach \i in {0,...,3} {
        \pgfmathtruncatemacro\idx{\i+1}
        \node[point,fill=red,label=180:{$u_\idx$}] (u\i) at (0,0.6*\i) {};
      }
      \foreach \j in {0,...,2} {
        \pgfmathtruncatemacro\idx{\j+1}
        \node[point,fill=blue,label=0:{$v_\idx$}] (v\j) at (3,0.9*\j) {};
      }
      \foreach \i in {0,...,3} {
        \foreach \j in {0,...,2} {
          \draw[arr] (u\i) -- (v\j) ;
        }
      }
      \draw[arr] (u3) -- node[above]{$a$} (v2);
    \end{tikzpicture}}
    \end{center}

    \vspace*{-\baselineskip}
    \centerline{(a)}
  \end{minipage}
  \begin{minipage}{0.5\textwidth}
    \begin{center}
    \scalebox{0.8}{\begin{tikzpicture}
      \tikzset{
        point/.style={circle,inner sep=0pt,minimum size=4pt},
        thickpoint/.style={circle,inner sep=0pt,minimum size=15pt},
        arr/.style={->,line width=0.1pt}
      }
      \node[point,fill=magenta,label=-90:{$u_\epsilon$}] (u) at (0,0) {};
      \node[point,fill=cyan,label=-90:{$v_\epsilon$}] (v) at (1.3,0) {};
      \draw[arr] (u) -- node[above]{$a$} (v);
      \foreach \i in {0,...,3} {
        \pgfmathtruncatemacro\idx{\i+1}
        \node[point,fill=red,label=180:{$u_\idx$}] (u\i) at ($(u) + (-1.3,-0.9+0.6*\i)$) {};
        \draw[arr] (u\i) -- node[above,pos=0.2]{$\epsilon$} (u);
      }
      \foreach \j in {0,...,2} {
        \pgfmathtruncatemacro\idx{\j+1}
        \node[point,fill=blue,label=0:{$v_\idx$}] (v\j) at ($(v) + (1.3,-0.9+\j*0.9)$) {};
        \draw[arr] (v\j) -- node[above,pos=0.2]{$\epsilon$} (v);
      }
    \end{tikzpicture}}
    \end{center}

    \vspace*{-\baselineskip}
    \centerline{(b)}
  \end{minipage}
  \vspace*{-1\baselineskip}
  \caption{ The complete bipartite graph $\protect\overrightarrow{K}_{4,3}$ (a) and
    one possible encoding using $\epsilon$-edges (b). }
  \label{fig:K43}
  \vspace*{-\baselineskip}
\end{figure}

%% file: definitions.tex
\section{Graphs}
\label{sec:graphs}

For simplicity reasons, in this paper we consider only networks whose
topologies are described by oriented binary graphs, where vertices
model processes and edges model a symmetric rendezvous between two
processes, \ie we do not distinguish the process initiating the
communication and consider that both participants have equal roles.

Let $\valpha$ and $\ealpha$ be finite disjoint alphabets of vertex and
edge labels, respectively. A \emph{graph} is a tuple
$\graph=(\vertof{\graph},\edgeof{\graph},\labof{\graph})$, where
$\vertof{\graph}$ is a finite set of vertices, $\edgeof{\graph}
\subseteq \vertof{\graph} \times \ealpha \times \vertof{\graph}$ is a
set of edges, such that $(v,\elab,v') \in \edgeof{\graph}$ implies
$v\neq v'$ (\ie graphs have no self-loops) and $\labof{\graph} :
\vertof{\graph} \rightarrow \pow{\valpha}$ assigns a set of labels
to each vertex. We do not distinguish graphs that are isomorphic.
We denote by $\graphs{\valpha}{\ealpha}$ the set of graphs with
vertex and edge labels from $\valpha$ and $\ealpha$, or simply
$\graphs{}{}$, if $\ealpha$ and $\valpha$ are understood.

We shall consider parameterized networks described as infinite sets of
graphs. Albeit infinite, these sets have a finite description, given
by finitely many inductive rules stating how the graphs in the set are
built from smaller graphs. This constructive\footnote{As
  opposed to the descriptive method of defining sets of graphs by
  their common properties, using \eg monadic second-order
  logic.} approach to the specification of infinite sets of graphs
relies on graph algebras, being at the core of an impressive body of
theoretical work (see \cite{courcelle_engelfriet_2012} for a survey).

\subsection{Algebras}
\label{sec:algebras}

We present two classical graph algebras, namely
\emph{vertex-replacement} (\vrtext) and \emph{hyperedge-replacement}
(\hrtext). Let $\valpha$ and $\ealpha$ be fixed in the following. We
consider a set $\palpha \subseteq \valpha$ of distinguished vertex
labels, called \emph{ports}, such that $\cardof{\labof{\graph}(v) \cap
  \palpha} \leq 1$, for each graph
$\graph\in\graphs{\valpha}{\ealpha}$ and each vertex $v \in
\vertof{\graph}$.
In other words, a vertex is labeled with at most one port.
We say that a vertex $v\in\vertof{\graph}$ is a $\plab$-port of $\graph$
whenever $\plab\in\labof{\graph}(v)$.
The \emph{sort} of a graph is the set of ports that occur in the labels
of the vertices from the graph.

The labels from $\valpha\setminus\palpha$ will be needed later to
associate vertices with local behaviors, referred to as process
types. For the time being, we assume the set of process types to be
$\valpha\setminus\palpha=\set{\vlab_1,\ldots,\vlab_n}$ and that the
set of port labels is partitioned accordingly, \ie $\palpha =
\palpha_{\vlab_1} \uplus \ldots \uplus \palpha_{\vlab_n}$. We say that
the \emph{process type} of a port $\plab\in\palpha$ is
$\vlab\in\valpha\setminus\palpha$, denoted $\ptypeof{}(\plab)=\vlab
\iffdef \plab \in \palpha_\vlab$. A partial function $\alpha : \palpha
\rightharpoonup \palpha$ is \emph{type-preserving} iff
$\alpha(\plab)$, if defined, has the same process type as $\plab$, for
each $\plab\in\palpha$.

\begin{figure}[t!]
  \vspace*{-\baselineskip}
  \centerline{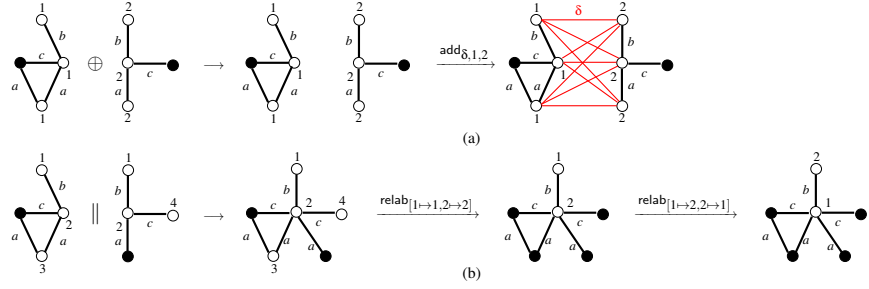}
  \caption{\vrtext{} operations (a) \hrtext{} operations (b). Ports
    are depicted as shallow circles; port labels are natural
    numbers. }
  \label{fig:graphs}
  \vspace*{-\baselineskip}
\end{figure}

The domain of the \emph{\vrtext{}-algebra} is the set $\graphs{}{}$ of
graphs. The operations of \vrtext{} are the following, see
\figref{fig:graphs} (a) for an illustration: \begin{compactitem}[-]
\item \emph{Disjoint union}: $\graph_1 \vrpop{\plabs_1}{\plabs_2}
  \graph_2 \isdef H$ iff the sorts of $\graph_1$ and $\graph_2$ are
  $\plabs_1$ and $\plabs_2$, respectively, $\vertof{H} =
  \vertof{\graph_1} \uplus \vertof{\graph_2}$, $\edgeof{H} =
  \edgeof{\graph_1} \uplus \edgeof{\graph_2}$ and $\labof{H} =
  \labof{\graph_1} \uplus \labof{\graph_2}$ where, if
  $\vertof{\graph_1} \cap \vertof{\graph_2} \neq \emptyset$ or
  $\edgeof{\graph_1} \cap \edgeof{\graph_2} \neq \emptyset$, we
  replace $\graph_2$ by an isomorphic copy disjoint from $\graph_1$.
  We write $\vrpop{}{}$ instead of $\vrpop{\plabs_1}{\plabs_2}$ when
  the argument sorts are understood.
\item \emph{Edge creation}: for an edge label $\elab\in\ealpha$ and
  port labels $\plab_1 \neq \plab_2\in\palpha$,
  $\addedge{\elab}{\plab_1}{\plab_2}(\graph) \isdef H$ iff
  $\vertof{H}=\vertof{\graph}$,
  $\edgeof{H}=\edgeof{\graph}\cup\set{(v_1,\elab,v_2) \mid v_1, v_2
    \in \vertof{\graph},~ \plab_i\in\labof{\graph}(v_i),~ i=1,2}$ and
  $\labof{H}=\labof{\graph}$; in other words,
  $\addedge{\elab}{\plab_1}{\plab_2}$ adds a directed $\elab$-edge
  from each $\plab_1$-port to each $\plab_2$-port of its argument. No
  edge is added if such an edge already exists. Moreover, because
  $\plab_1$ and $\plab_2$ are different labels and each vertex has at
  most one port label, no self-loops are introduced by this operation.
\item \emph{Port redefinition}: for a type-preserving partial function
  $\alpha : \palpha \rightharpoonup \palpha$, $\relab{\alpha}(\graph)
  \isdef H$ iff $\vertof{H}=\vertof{\graph}$,
  $\edgeof{H}=\edgeof{\graph}$ and $\labof{H}(v) = \{\alpha(\plab)
  \mid \plab \in \labof{\graph}(v) \cap \palpha\} \cup
  (\labof{\graph}(v) \setminus \palpha)$, for each $v\in\vertof{H}$,
  \ie the port labels are redefined according to $\alpha$ (if
  $\alpha(\plab)$ is undefined, the port label is erased) and the
  other labels are unchanged.
\item \emph{Single-vertex graphs}: for a port label $\plab\in\palpha$,
  $\svertex{\plab} \isdef H$ iff $\vertof{H}=\set{u}$ for some vertex
  $u$, $\edgeof{H}=\emptyset$ and $\labof{H}(u)=\set{\plab,\ptypeof{}(\plab)}$. 
\end{compactitem}
For simplicity, we use the same notation for a \vrtext{} operation
above and its corresponding function symbol. Since $\ealpha$ and
$\palpha$ are finite sets, the above set of function symbols is
finite. A \vrtext-term $\theta[x_1,\ldots,x_n]$ consists of function
symbols (\ie the binary $\vrpop{}{}$ and unary
$\addedge{\elab}{\plab_1}{\plab_2}$ and $\relab{\alpha}$) and the
variables $x_1,\ldots,x_n$ of arity zero. A term is ground if it has
no occurrences of variables. We write $\val{\vr}{\theta}$ for the graph obtained
by interpreting the function symbols in the ground term $\theta$ as the
corresponding \vrtext{}-operations. 

\begin{example}\label{ex:K43VR}
  The \vrtext-term \( \theta \isdef
  \relab{\emptyset}(\addedge{a}{\plab}{\plab'}(
  (\bigvrpop{i\in 1,...,4} \svertex{\plab})
  \vrpop{}{}
  (\bigvrpop{j\in 1,...,3} \svertex{\plab'}) )) \)
  evaluates to the complete
  bipartite graph $\val{\vr}{\theta} = \overrightarrow{K}_{4,3}$ with
  directed edges $a$, see \figref{fig:K43} (a).
  Edge labels are $\ealpha = \set{a}$,
  port labels are $\palpha = \set{\plab,\plab'}$,
  and vertex labels are
  $\valpha = \palpha\cup\set{\ptypeof{}(\plab),\ptypeof{}(\plab')}$.
  The $\plab$-ports are $u_1,...,u_4$ and the $\plab'$-ports are $v_1,...,v_3$
  in \figref{fig:K43} (a).
  We write $\emptyset$ for the empty partial function, that removes all port labels.
\end{example}


The domain of the \emph{\hrtext{}-algebra} is the set
$\hrgraphs{\valpha}{\ealpha} \subseteq \graphs{\valpha}{\valpha}$ of
graphs $\graph$, written $\hrgraphs{}{}$ when $\valpha$ and $\ealpha$
are understood, where $\labof{\graph}(v_1) \cap \labof{\graph}(v_2)
\cap \palpha \neq \emptyset$ implies $v_1 = v_2$, for any two vertices
$v_1,v_2 \in \vertof{\graph}$. In other words, a port labels at most
one vertex. In the usual terminology, ports are called \emph{sources},
when the \hrtext{}-algebra is understood from the context. The
operations of \hrtext{} are the following, see \figref{fig:graphs} (b)
for an illustration: \begin{compactitem}[-]
\item \emph{Composition}: $\graph_1 \hrpop{\plabs_1}{\plabs_2} \graph_2
  \isdef H$ iff $H$ is obtained from the disjoint union of $\graph_1,
  \graph_2 \in \hrgraphs{}{}$ of sorts $\plabs_1$ and $\plabs_2$,
  respectively (taking an isomorphic copy of $\graph_2$ disjoint from
  $\graph_1$, if necessary) by joining all pairs of vertices $v_i \in
  \vertof{\graph_i}$, for $i=1,2$ such that $\labof{\graph_1}(v_1)
  \cap \labof{\graph_2}(v_2) \cap \palpha \neq \emptyset$. The label
  of a vertex $v$ obtained by joining $v_1 \in \vertof{\graph_1}$ with
  $v_2 \in \vertof{\graph_2}$ is $\labof{\graph_1}(v_1) \cup
  \labof{\graph_2}(v_2)$. Since the label of a vertex in a graph from
  $\hrgraphs{}{}$ contains at most one source, the composition of any
  two graphs from $\hrgraphs{}{}$ is contained in $\hrgraphs{}{}$. We
  write $\hrpop{}{}$ instead of $\hrpop{\plabs_1}{\plabs_2}$ when
  the argument sorts are understood. 
\item \emph{Source redefinition}: $\relab{\alpha}$ is defined in the
  same way as for \vrtext{}, with the further restriction that
  $\alpha$ is an injective partial function. This restriction is
  necessary to ensure that $\relab{\alpha}(\graph)\in\hrgraphs{}{}$ if
  $\graph\in\hrgraphs{}{}$.
\item \emph{Single-vertex graphs}: are defined in the same way as for
  \vrtext{}.
\item \emph{Single-edge graphs}: for an edge label $\elab\in\ealpha$
  and two port labels $\plab_1 \neq \plab_2\in\palpha$,
  $\edge{\elab}{\plab_1}{\plab_2}\isdef H$ iff
  $\vertof{H}=\set{v_1,v_2}$, $\edgeof{H}=\set{(v_1,\elab,v_2)}$,
  $\labof{H}(v_i)=\set{\plab_i,\ptypeof{}(\plab_i)}$, for $i=1,2$.
\end{compactitem}
\hrtext{}-terms are defined just as \vrtext{}-terms and $\val{\hr}{\theta}$
denotes the graph obtained by interpreting the function symbols in the
ground term $\theta$ as the corresponding \hrtext{}-operations.

\begin{example}\label{ex:K43HR}
  We write $(\plab)$ for the partial function that acts as the
  identity on $\set{\plab}$, erasing all source labels except $\plab$.
  The \hrtext-term 
  \( \theta' \isdef
    \relab{\emptyset}(\edge{a}{\plab_\epsilon}{\plab'_\epsilon})
    \hrpop{}{}
    (\bighrpop{i\in 1,...,4}
      \relab{(\plab_\epsilon)}(\edge{\epsilon}{\plab}{\plab_\epsilon})
    )
    \hrpop{}{}
    (\bighrpop{j\in 1,...,3}
      \relab{(\plab'_\epsilon)}(\edge{\epsilon}{\plab'}{\plab'_\epsilon})
    )
  \)
  evaluates to the graph shown in \figref{fig:K43} (b).  The edge
  labels are $\ealpha = \set{a,\epsilon}$, the port labels are
  $\palpha = \set{\plab,\plab',\plab_\epsilon,\plab'_\epsilon}$, the
  vertex labels are $\valpha = \palpha \cup \{\ptypeof{}(\plab),
  \ptypeof{}(\plab'), \ptypeof{}(\plab_\epsilon),
  \ptypeof{}(\plab'_\epsilon)\}$. The $\plab$, $\plab'$,
  $\plab_\epsilon$ and $\plab'_\epsilon$-ports in \figref{fig:K43}
  are respectively $u_{1,...,4}$, $v_{1,...,3}$, $u_\epsilon$, $v_\epsilon$.
\end{example}


\subsection{Grammars}
\label{sec:grammars}

Graph grammars are a standard way to represent sets of graphs.
Let $\Omega$ be an algebra, either \vrtext{} or \hrtext{}.
A $\Omega$-\emph{grammar} is a pair $\grammar=(\nonterm,\rules)$
consisting of a finite set $\nonterm$ of \emph{nonterminals}
and a finite set $\rules$ of \emph{rules} of the form,
either (1) $X \rightarrow t[X_1,\ldots,X_n]$,
where $X,X_1,\ldots,X_n \in \nonterm$ are nonterminals and $t$
is a $\Omega$-term whose only variables are $X_1,\ldots,X_n$,
or (2) $\rightarrow X$, where $X \in \nonterm$;
the rules of this form are called \emph{axioms}.
Given $\Omega$-terms $\theta$ and $\eta$, a \emph{step}
$\theta \step{\grammar} \eta$ obtains $\eta$ from $\theta$
by replacing an occurrence of a nonterminal $X$ with the term $t$,
for some rule $X \rightarrow t[X_1,\ldots,X_n]$ of $\grammar$.
A $X$-\emph{derivation} is a sequence of steps starting with a nonterminal $X$.
The derivation is \emph{complete} if it ends with a ground term.
Let $\alangof{X}{\grammar} \isdef \{\val{\Omega}{\theta}
  \mid X \step{\grammar}^* \theta \text{ is a complete derivation}\}$
and $\alangof{}{\grammar} \isdef \bigcup_{\rightarrow X \in \rules}
  \alangof{X}{\grammar}$
be the \emph{language} of $\grammar$, the algebra $\Omega$
being understood from the rules of $\grammar$.
A set of graphs is \vrtext{} (\resp \hrtext{})
iff it is the language of a \vrtext{} (\resp \hrtext{}) grammar.

\begin{example}\label{ex:grammars}
  The \vrtext-grammar $\grammar$ below produces the term
  $\theta$ from \exref{ex:K43VR}.
  The language of this grammar contains \figref{fig:K43} (a).
  The \hrtext-grammar $\grammar'$ produces the term $\theta'$ from
  \exref{ex:K43HR}.
  The language of this grammar contains \figref{fig:K43} (b). \\
\begin{minipage}{0.5\textwidth}
  \begin{center}
  $\grammar: \left\{
  \begin{array}{rl}
      \to &~ S \\
    S \to &~ \relab{\emptyset}{} (\addedge{a}{\plab}{\plab'}{}(K)) \\
    K \to &~ \svertex{\plab} \\
    K \to &~ \svertex{\plab'} \\
    K \to &~ K \vrpop{}{} K
  \end{array}\right.$
  \end{center}
\end{minipage}
\begin{minipage}{0.5\textwidth}
  \begin{center}
    $\grammar': \left\{
    \begin{array}{rl}
      \to &~ S \\
    S \to &~ \relab{\emptyset} (\edge{a}{\plab_\epsilon}{\plab'_\epsilon} \hrpop{}{} K) \\
    K \to &~ \relab{({\plab_\epsilon})} (\edge{\epsilon}{\plab}{\plab_\epsilon}) \\
    K \to &~ \relab{({\plab'_\epsilon})} (\edge{\epsilon}{\plab'}{\plab'_\epsilon}) \\
    K \to &~ K \hrpop{}{} K
  \end{array}\right.$
  \end{center}
\end{minipage}
\end{example}

It is known that the expressivity of \vrtext-grammars strictly
subsumes that of \hrtext-grammars~\cite{ENGELFIET1990163}.
The relation between \vrtext{} and \hrtext{} is made precise by the following:

\begin{definition}\label{def:expansion}
  Let $\expalpha$ be a set of edge labels disjoint from $\ealpha$. An
  $\expalpha$-graph is a graph
  $\graph \in \graphs{\valpha}{\ealpha \cup \expalpha \cup \protect\overleftarrow{\expalpha}}$
  such that \begin{compactitem}[-]
  \item $\overleftarrow{\expalpha} \isdef \set{\overleftarrow{e} \mid e \in \expalpha}$,
  and for every edge $(u, e, v)$, $e \in \expalpha$,
  there exists an edge $(v, \overleftarrow{e}, u)$.
  \item the subgraph of $\graph$ consisting
  of $\expalpha$-labeled edges and the vertices that have incoming
  $\expalpha$-labeled edges is a forest,
  \ie every non-root vertex points to a single parent.
  \end{compactitem}
  The \emph{expansion} of an $\expalpha$-graph $\graph\in
  \graphs{\valpha}{\ealpha \cup \expalpha \cup \protect\overleftarrow{\expalpha}}$
  is $\expof{\graph} \in \graphs{\valpha}{\ealpha}$, where: \begin{compactitem}[-]
  \item $\vertof{\expof{\graph}}$ is the set of vertices of $\graph$
    that are not the target of an $\expalpha$-labeled edge,
\item $\edgeof{\expof{\graph}}$ is the set of edges $(v_1,\elab,v_2)$
  for which there exist $\expalpha$-labeled paths from $v_i$ to some
  vertices $v'_i \in \graph$, for $i=1,2$, such that
  $(v'_1,\elab,v'_2) \in \edgeof{\graph}$, for some $\elab\in\ealpha$,
\item $\labof{\expof{\graph}}(v) \isdef \labof{\graph}(v)$.
\end{compactitem}
\end{definition}

\begin{proposition}[Proposition 2.4 in \cite{COURCELLE1995275}]
  For each \vrtext-grammar $\grammar$, one can build a \hrtext-grammar
  $\grammar'$ such that
  $\alangof{}{\grammar}=\expof{\alangof{}{\grammar'}}$.
\end{proposition}

\begin{example}\label{ex:expansion}
  The \hrtext-grammar $\grammar'$ from \exref{ex:grammars} is
  obtained by transformation of the \vrtext-grammar $\grammar$
  from the same example, such that
  $\alangof{}{\grammar}=\expof{\alangof{}{\grammar'}}$, where
  $\expalpha=\set{\epsilon}$.
  We omit $\overleftarrow{\expalpha}$-edges from this figure:
  their placement is simply the inverse of $\expalpha$-edges.
\end{example}

\section{Parameterized Systems}
\label{sec:systems}

We aim at describing a family of networks by a \vrtext-grammar. To
model the behavior of a network, we associate to each vertex in the graph
a \emph{process type}, which is a finite-state machine represented as
a Petri net of a particular form. We denote by $\ptypes$ the fixed and
finite set of process types. In the rest of this paper, the alphabet
of vertex labels is assumed to be $\valpha \isdef \ptypes \uplus
\palpha$, where $\palpha=\biguplus_{\ptype\in\ptypes}\palpha_\ptype$
is a set of port labels partitioned according to the process types,
\ie $\palpha_\ptype$ is the set of ports corresponding to the process
type $\ptype$. The intuition is that different vertices labeled with a
process type run copies of that process type, called
\emph{processes}. Neighbouring processes communicate by joining their
transitions in a symmetric synchronous rendezvous, represented by a
pair of transitions, that labels the edge between their host
vertices. Hence, the alphabet $\ealpha$ of edge labels is considered
to be the set of pairs of transitions from the process types
$\ptypes$.

\subsection{Petri Nets}
\label{subsec:pn}

We make the definition of process types precise by recalling Petri
nets. A \emph{net} is a tuple $\anet =
(\placeof{\anet},\transof{\anet},\weightof{\anet})$, where
$\placeof{\anet}$ is a finite set of \emph{places}, $\transof{\anet}$
is a finite set of \emph{transitions}, disjoint from
$\placeof{\anet}$, and $\weightof{\anet} : (\placeof{\anet} \times
\transof{\anet}) \cup (\transof{\anet} \times \placeof{\anet})
\rightarrow \nat$ is a \emph{weighted incidence relation} between
places and transitions. For all $x,y\in\placeof{\anet} \cup
\transof{\anet}$ such that $\weightof{\anet}(x,y) > 0$, we say that
there is an \emph{edge of weight} $\weightof{\anet}(x,y)$ between $x$
and $y$. For an element $x\in\placeof{\anet} \cup \transof{\anet}$, we
define the set of \emph{predecessors} $\pre{x} \isdef \set{y \in
  \placeof{\anet} \cup \transof{\anet} \mid \weightof{\anet}(y,x)>0}$,
\emph{successors} $\post{x} \isdef \set{y \in \placeof{\anet} \cup
  \transof{\anet} \mid \weightof{\anet}(x,y)>0}$ and
predecessor-successor pair $\prepost{x} \isdef (\pre{x},\post{x})$.
If not obvious from the context, we will specify the net in which the predecessor
and successor and considered: $(\pre{x})_\anet$, $(\post{x})_\anet$, $(\prepost{x})_\anet$.

A \emph{marking} of $\anet$ is a function $\amark : \placeof{\anet}
\rightarrow \nat$. A transition $t$ is \emph{enabled} in the marking
$\amark$ if $\amark(q) \geq \weightof{\anet}(q,t)$, for each place $q
\in \places$. For all markings $\amark$, $\amark'$ and transitions $t
\in \trans$, we write $\amark \fire{t} \amark'$ whenever $t$ is
enabled in $\amark$ and $\amark'(q) = \amark(q) -
\weightof{\anet}(q,t) + \weightof{\anet}(t,q)$, for all $q \in
\placeof{\anet}$.
Given a marking $\amark$, a sequence of transitions
$\rho = (t_1,t_2,...) \in (\transof{\amarkednet})^\infty$
is a \emph{firing sequence} iff either \begin{inparaenum}[(i)]
\item $\rho$ is empty, or
\item there exist markings $\amark_1,\amark_2,\ldots$
  such that $\amark \fire{t_1} \amark_1 \fire{t_2} \amark_2 \fire{t_3} \cdots$
\end{inparaenum}.
When $\rho$ has finite length $n$, we can write $\amark \fire{\rho} \amark_n$,
and we say that $\rho$ is a firing sequence from $\amark$ to $\amark_n$.

A \emph{Petri net} (PN) is a pair $\amarkednet=(\anet,\amark_0)$,
where $\anet$ is a net and $\amark_0$ is the \emph{initial marking} of
$\anet$. For simplicity, we write $\placeof{\amarkednet} \isdef
\placeof{\anet}$, $\transof{\amarkednet} \isdef \transof{\anet}$,
$\weightof{\amarkednet} \isdef \weightof{\anet}$ and
$\initmarkof{\amarkednet}\isdef\amark_0$ for the elements of
$\amarkednet$. A marking $\amark$ is \emph{reachable} in $\amarkednet$
iff there exists a finite firing sequence $\rho$ such that $\amark_0
\fire{\rho} \amark$.
We denote by $\pathsof{\amarkednet}$ the set of finite or infinite
firing sequences of $\amarkednet$ starting from $\initmarkof{\amarkednet}$.

\subsection{Behaviors}
\label{subsec:behaviors}

We formalize the behavior of a system (\ie a graph whose vertices are
labeled by process types) using PNs. To do so, we associate to each
process type a PN having a special form:

\begin{definition}\label{def:process-type}
  A \emph{process type} $\ptype$ is a PN having weights at most
  $1$ and exactly one marked place initially, whose transitions are
  partitioned into \emph{observable} $\obstransof{\ptype}$ and
  \emph{internal} $\inttransof{\ptype}$, \ie
  $\transof{\ptype}=\obstransof{\ptype} \uplus \inttransof{\ptype}$,
  such that each transition has exactly one predecessor and one
  successor. Let $\ptypes = \set{\ptype_1, \ldots, \ptype_k}$ be a
  finite fixed set of process types such that $\placeof{\ptype_i} \neq
  \emptyset$, for all $i \in \interv{1}{k}$ and $\placeof{\ptype_i}
  \cap \placeof{\ptype_j} = \emptyset$, for all $1 \leq i < j \leq k$.
\end{definition}
Because a process type has exactly one initial token and all
transitions have one predecessor and one successor, every reachable
marking of a process type has exactly one token.
We denote by $\placeof{\ptypes}\isdef\biguplus_{\ptype\in\ptypes}
\placeof{\ptype}$ and $\obstransof{\ptypes} \isdef
\biguplus_{\ptype\in\ptypes} \obstransof{\ptype}$ the sets of places
and observable transitions from some $\ptype\in\ptypes$, respectively.

\begin{example}\label{ex:proctypes}
  \figref{fig:proctype} (a) shows two examples of process types, $Once$ and $Loop$.
  Observable transitions are shown in black, and internal transitions in yellow.
  They have $\placeof{Once} = \set{\mathrm{on}, \mathrm{off}}$,
  $\obstransof{Once} = \transof{Once} = \set{\mathrm{send}}$,
  $\placeof{Loop} = \set{\mathrm{free}, \mathrm{busy}}$,
  $\obstransof{Loop} = \set{\mathrm{recv}}$
  and $\inttransof{Loop} = \set{\mathrm{handle}}$.
\end{example}

\input{figure-proctype}

\begin{definition}\label{def:system}
  A \emph{system} $\asys=(\vertof{\asys},\edgeof{\asys},\labof{\asys})
  \in \graphs{\ptypes \uplus
    \palpha}{\obstransof{\ptypes}\times\obstransof{\ptypes}}$ is a
  graph whose vertices $v \in \vertof{\asys}$ are labeled with exactly
  one process type $\procof{\asys}(v)=\ptype \iffdef
  \labof{\asys}(v)\cap\ptypes=\set{\ptype}$ and at most one port
  $\portof{\asys}(v)=\plab \iffdef
  \labof{\asys}(v)\cap\palpha=\set{\plab}$.  Edges
  $v_1\arrow{(t_1,t_2)}{} v_2 \in \edgeof{\asys}$ are labeled with
  pairs of observable transitions, such that $t_i \in
  \obstransof{\procof{\asys}(v_i)}$, for $i=1,2$. We denote by
  $\systems{\ptypes}\isdef\graphs{\ptypes}{\obstransof{\ptypes}\times\obstransof{\ptypes}}$
  the set of systems of empty sort, \ie without ports.
\end{definition}

The communication (\ie synchronization between processes) in a system
is formally captured by the following notion of behavior:

\begin{definition}\label{def:behavior}
  A \emph{behavior} is a PN $\amarkednet$ such that $1 \leq
  \cardof{\pre{t}}=\cardof{\post{t}} \leq 2$, for each
  $t\in\transof{\amarkednet}$. The \emph{behavior of a system}
  $\asys$ is $\behof{\asys} \isdef
  (\anet,\amark_0)$, where: \begin{itemize}
  \item $\placeof{\anet} \isdef \set{(q,v) \mid q \in
    \placeof{\procof{\asys}(v)},~ v \in \vertof{\asys}}$, a place
    $(q,v)$ corresponds to the place $q$ of the process type that
    labels the vertex $v$,
  \item $\transof{\anet} \isdef \edges \cup \set{(t,v) \mid t \in
    \inttransof{\procof{\asys}(v)},~ v \in \vertof{\asys}}$, the
    transitions are either edges of the system (\ie modeling the
    synchronizations of two processes) or pairs $(t,v)$ corresponding
    to an internal transition $t$ of the process type that labels the
    vertex $v$,
  \item the weight function $\weightof{\anet}$ is such that
     internal transitions are preserved, and observable transitions
     are merged according to edge labels.
      That is, for every $v$ such that
      $(\prepost{t})_{\procof{\asys}(v)} = (\set{q}, \set{q'})$
      we have $(\prepost{(t,v)})_\anet = (\set{(q,v)}, \set{(q',v')})$,
      and for every edge $e = (v \arrow{(t,t')}{} v')$
      with $(\prepost{t})_{\procof{\asys}(v)} = (\set{q_1},\set{q_2})$
      and $(\prepost{t'})_{\procof{\asys}(v')} = (\set{q'_1},\set{q'_2})$
      we have $(\prepost{e})_\anet = (\set{(q_1,v),(q'_1,v')},\set{(q_2,v),(q'_2,v')})$.
  \item $\amark_0(q,v) \isdef \initmarkof{\procof{\asys}(v)}(q)$, for
    all $v \in \vertof{\asys}$ and $q \in
    \placeof{\procof{\asys}(v)}$.
  \end{itemize}
\end{definition}

\begin{example}\label{ex:systems}
  \figref{fig:proctype} (b) shows the behavior of
  $\overrightarrow{K}_{4,3}(\mathrm{send},\mathrm{recv})$, having the
  architecture shown in \figref{fig:K43} (a), with the red vertices
  $u_{1\ldots4}$ labeled with the process type $\mathit{Once}$ and
  blue vertices $v_{1\ldots3}$ labeled with the process type
  $\mathit{Loop}$, from \figref{fig:proctype}. The edges are labeled
  by the pair $(\mathrm{send},\mathrm{recv})$ of observable
  transitions from $\mathit{Once}$ and $\mathit{Loop}$.
\end{example}


\subsection{The Parametric Reachability Problem}

A \emph{parametric} \vrtext{} (\resp \hrtext{}) \emph{system} is
described by a \vrtext{} (\resp \hrtext{}) grammar $\grammar$. We
define our decision problem as a parametric reachability problem
asking if there exists an instance of a parametric system, described
by a grammar, whose behavior reaches an error configuration from a
given set. If the answer is negative, we have a proof that the
parametric system is correct.

Let $\vars$ be a fixed countably infinite set of variables,
interpreted over natural numbers. An arithmetic formula $\alpha$ is a
(possibly quantified) first-order formula using variables in $\vars$,
constants denoting natural numbers, the binary operations of addition
and multiplication and the (in)equality relation. We denote by
$\atoms$ the set of such formul{\ae}. Let $\amarkednet$ be
a PN and $\varlab : \placeof{\amarkednet} \rightharpoonup \vars$ be a
partial function that labels certain places of $\amarkednet$ with
variables. The boolean value $\sem{\alpha}_\amark^\varlab \in
\set{\mathrm{true},\mathrm{false}}$ of an arithmetic formula $\alpha$
in a marking $\amark : \placeof{\amarkednet} \rightarrow \nat$ is
obtained by replacing each variable $x$ that occurs free in $\alpha$
by the integer value $\sum_{q \in \varlab^{-1}(x)} \amark(q)$, \ie the
total number of tokens from the places associated with $x$ in
$\amark$. We further define $\atomsof{\amark}{\varlab} \isdef
\set{\alpha \in \atoms \mid \sem{\alpha}_\amark^\varlab =
  \mathrm{true}}$, \ie the set of arithmetic formul{\ae} satisfied by
a given marking and a variable labeling.  A PN $\amarkednet$ satisfies
a reachability specification $\alpha$ for a place labeling $\varlab$,
written $(\amarkednet,\varlab) \models \alpha$, iff there exists a
finite firing sequence $\rho$ leading to a reachable marking
$\initmarkof{\amarkednet} \fire{\rho} \amark'$ such that $\alpha \in
\atomsof{\amark'}{\varlab}$.

We specialize the satisfiability of a reachability specification
by an arbitrary PN to the satisfiability by the behavior of
a system $\asys \in \systems{\ptypes}$ and a labeling
$\syslab : \placeof{\ptypes} \rightharpoonup \vars$ of the places
from the process types $\ptypes$ with variables:
\begin{align*}
  (\asys,\syslab) \models \alpha \iffdef (\behof{\asys},\syslab \circ \proj{}{1}) \models \alpha
\end{align*}
where $\proj{}{1}$ is the projection of a pair on its first component,
\ie $\syslab \circ \proj{}{1}$ labels each pair $(q,v) \in
\placeof{\behof{\asys}}$ by the variable $\syslab(q)$, if the latter
is defined.
Intuitively, we require of $\varlab$ that it assigns the same variable
to all ``copies'' of the same place $\set{(q,v) \mid v \in \vertof{\asys}}$.
Hence, the boolean value $\sem{\alpha}_\amark^{\syslab
  \circ \proj{}{1}}$ of an atomic proposition in a marking $\amark$ of
$\behof{\asys}$ is obtained by mapping each free variable $x$ from
$\alpha$ to the total number of tokens from a place $(q,v) \in
\placeof{\behof{\asys}}$ such that $\syslab(q)=x$. This paper is
concerned with the following problem:

\begin{definition}\label{def:pmcp}
  Let $\ptypes$ be a set of process types. The Parametric Reachability Problem
  $\prp{\grammar}{\alpha}{\syslab}{\ptypes}$ takes as input a
  grammar $\grammar$ such that $\alangof{}{\grammar} \subseteq
  \systems{\ptypes}$, a reachability formula $\alpha$ and a
  place labeling $\syslab : \placeof{\ptypes} \rightharpoonup \vars$,
  and asks if there exists some system $\asys \in \alangof{}{\grammar}$
  for which $(\asys,\syslab) \models \alpha$ holds.
\end{definition}

\begin{example}
  Choosing $\grammar$ the \vrtext-grammar (\exref{ex:grammars}),
  the set of variables $\vars \isdef \set{x,y}$, the place labeling
  $\syslab = [\mathrm{on} \mapsto x, \mathrm{off} \mapsto y]$
  defined on the process types from \figref{fig:proctype} (a), and the
  formula $\alpha \isdef y \geq x + 2$, we find that
  $\prp{\grammar}{\alpha}{\syslab}{\ptypes}$ has a positive answer,
  because the behavior \figref{fig:proctype} (b) which belongs to
  $\behof{\langu(\grammar)}$ admits a reachable marking $\amark$
  with $3$ tokens on $\mathrm{off}$ and $1$ token on $\mathrm{on}$,
  satisfying $\sem{\alpha}_\amark^{\syslab\circ\proj{}{1}} =
  \sem{3 \geq 1 + 2} = \mathrm{true}$. 
\end{example}

In the light of inherent theoretical boundaries, related to the
undecidability of the above parametric verification
problem~\cite{DBLP:conf/popl/EmersonN95}, several semi-algorithmic
methods have been proposed, for parametric \hrtext{}
systems~\cite{arXiv2025} and parametric systems described using
similar formalisms, such as a dialect of separation logic with
inductive definitions~\cite{DBLP:journals/tcs/BozgaIS23} and a
recursive term algebra~\cite{DBLP:conf/facs2/BozgaI21}. In particular,
the architecture description languages used in
\cite{DBLP:journals/tcs/BozgaIS23,DBLP:conf/facs2/BozgaI21} have
similar graph composition and relabeling as the more standard
\hrtext{} grammars. For these methods, the verification of certain
coverability properties (\eg absence of deadlocks and critical section violations)
relies on the generation of structural invariants for the parametric
family of behaviors, directly from the specification of the
architectures and the process types. Examples include \emph{trap
  invariants}, \ie sets of places that, once marked, remain forever
marked, and \emph{mutex invariants}, \ie sets of places of which at
most one is marked. The generation of structural invariants can be
redefined for \hrtext{} grammars at little cost.

In contrast, there is virtually no verification method for parametric
\vrtext{} systems. Many decidability results in the literature
consider parametric systems with clique
architectures~\cite{GermanSistla92}, that can be described by simple
\vrtext-grammars. A prominent result is the decidability of the more
general parametric model checking problem ($\pmcp{}{}{}{}$) for
\emph{token-passing systems} (where communication is restricted to the
passing of a token between processes) with \mso-definable bounded
clique-width architectures~\cite[Theorem
  26]{DBLP:conf/concur/AminofKRSV14}. It is known that each
\mso-definable set of bounded clique-width is definable by a \vrtext{}
grammar, but not viceversa.
Moreover, an orthogonal\footnote{This result applies to all \hrtext{}
sets of architectures, not just the \mso-definable ones.}
decidability result for parametric \hrtext{} token-passing systems is
given in \cite[Theorem 4]{arXiv2025}. It is an interesting open
problem whether this decidability result for \hrtext{} systems carries
over to \vrtext{} systems, but this exceeds the scope of the
current paper, and will be considered in future work.

%% file: graphs.pdf_t
\begin{picture}(0,0)%
\includegraphics{graphs.pdf}%
\end{picture}%
%
%
\setlength{\unitlength}{1776sp}%
\begingroup\makeatletter\ifx\SetFigFont\undefined%
\gdef\SetFigFont#1#2#3#4#5{%
  \reset@font\fontsize{#1}{#2pt}%
  \fontfamily{#3}\fontseries{#4}\fontshape{#5}%
  \selectfont}%
\fi\endgroup%
\begin{picture}(11907,3954)(811,-739)
\put(7201,1214){\makebox(0,0)[b]{\smash{{\SetFigFont{6}{7.2}{\rmdefault}{\mddefault}{\updefault}{\color[rgb]{0,0,0}(a)}%
}}}}
\put(3602,147){\makebox(0,0)[b]{\smash{{\SetFigFont{5}{6.0}{\rmdefault}{\mddefault}{\updefault}{\color[rgb]{0,0,0}$\longrightarrow$}%
}}}}
\put(6601,164){\makebox(0,0)[b]{\smash{{\SetFigFont{5}{6.0}{\rmdefault}{\mddefault}{\updefault}{\color[rgb]{0,0,0}$\arrow{\relab{[1\mapsto1,2\mapsto2]}}{}$}%
}}}}
\put(10201,164){\makebox(0,0)[b]{\smash{{\SetFigFont{5}{6.0}{\rmdefault}{\mddefault}{\updefault}{\color[rgb]{0,0,0}$\arrow{\relab{[1\mapsto2,2\mapsto1]}}{}$}%
}}}}
\put(8416,974){\makebox(0,0)[b]{\smash{{\SetFigFont{5}{6.0}{\rmdefault}{\mddefault}{\updefault}{\color[rgb]{0,0,0}$1$}%
}}}}
\put(8271,-192){\makebox(0,0)[b]{\smash{{\SetFigFont{5}{6.0}{\rmdefault}{\mddefault}{\updefault}{\color[rgb]{0,0,0}$a$}%
}}}}
\put(8267,506){\makebox(0,0)[b]{\smash{{\SetFigFont{5}{6.0}{\rmdefault}{\mddefault}{\updefault}{\color[rgb]{0,0,0}$b$}%
}}}}
\put(8715, 60){\makebox(0,0)[b]{\smash{{\SetFigFont{5}{6.0}{\rmdefault}{\mddefault}{\updefault}{\color[rgb]{0,0,0}$c$}%
}}}}
\put(8082,292){\makebox(0,0)[b]{\smash{{\SetFigFont{5}{6.0}{\rmdefault}{\mddefault}{\updefault}{\color[rgb]{0,0,0}$c$}%
}}}}
\put(7707,-135){\makebox(0,0)[b]{\smash{{\SetFigFont{5}{6.0}{\rmdefault}{\mddefault}{\updefault}{\color[rgb]{0,0,0}$a$}%
}}}}
\put(8581,331){\makebox(0,0)[b]{\smash{{\SetFigFont{5}{6.0}{\rmdefault}{\mddefault}{\updefault}{\color[rgb]{0,0,0}$2$}%
}}}}
\put(8613,-282){\makebox(0,0)[b]{\smash{{\SetFigFont{5}{6.0}{\rmdefault}{\mddefault}{\updefault}{\color[rgb]{0,0,0}$a$}%
}}}}
\put(12016,974){\makebox(0,0)[b]{\smash{{\SetFigFont{5}{6.0}{\rmdefault}{\mddefault}{\updefault}{\color[rgb]{0,0,0}$2$}%
}}}}
\put(11871,-192){\makebox(0,0)[b]{\smash{{\SetFigFont{5}{6.0}{\rmdefault}{\mddefault}{\updefault}{\color[rgb]{0,0,0}$a$}%
}}}}
\put(11867,506){\makebox(0,0)[b]{\smash{{\SetFigFont{5}{6.0}{\rmdefault}{\mddefault}{\updefault}{\color[rgb]{0,0,0}$b$}%
}}}}
\put(12315, 60){\makebox(0,0)[b]{\smash{{\SetFigFont{5}{6.0}{\rmdefault}{\mddefault}{\updefault}{\color[rgb]{0,0,0}$c$}%
}}}}
\put(11682,292){\makebox(0,0)[b]{\smash{{\SetFigFont{5}{6.0}{\rmdefault}{\mddefault}{\updefault}{\color[rgb]{0,0,0}$c$}%
}}}}
\put(11307,-135){\makebox(0,0)[b]{\smash{{\SetFigFont{5}{6.0}{\rmdefault}{\mddefault}{\updefault}{\color[rgb]{0,0,0}$a$}%
}}}}
\put(12181,331){\makebox(0,0)[b]{\smash{{\SetFigFont{5}{6.0}{\rmdefault}{\mddefault}{\updefault}{\color[rgb]{0,0,0}$1$}%
}}}}
\put(12213,-282){\makebox(0,0)[b]{\smash{{\SetFigFont{5}{6.0}{\rmdefault}{\mddefault}{\updefault}{\color[rgb]{0,0,0}$a$}%
}}}}
\put(1956,155){\makebox(0,0)[b]{\smash{{\SetFigFont{8}{9.6}{\rmdefault}{\mddefault}{\updefault}{\color[rgb]{0,0,0}$\hrpop{}{}$}%
}}}}
\put(1208,955){\makebox(0,0)[b]{\smash{{\SetFigFont{5}{6.0}{\rmdefault}{\mddefault}{\updefault}{\color[rgb]{0,0,0}$1$}%
}}}}
\put(1490,558){\makebox(0,0)[b]{\smash{{\SetFigFont{5}{6.0}{\rmdefault}{\mddefault}{\updefault}{\color[rgb]{0,0,0}$b$}%
}}}}
\put(2414,956){\makebox(0,0)[b]{\smash{{\SetFigFont{5}{6.0}{\rmdefault}{\mddefault}{\updefault}{\color[rgb]{0,0,0}$1$}%
}}}}
\put(3042,331){\makebox(0,0)[b]{\smash{{\SetFigFont{5}{6.0}{\rmdefault}{\mddefault}{\updefault}{\color[rgb]{0,0,0}$4$}%
}}}}
\put(2269,-210){\makebox(0,0)[b]{\smash{{\SetFigFont{5}{6.0}{\rmdefault}{\mddefault}{\updefault}{\color[rgb]{0,0,0}$a$}%
}}}}
\put(2265,488){\makebox(0,0)[b]{\smash{{\SetFigFont{5}{6.0}{\rmdefault}{\mddefault}{\updefault}{\color[rgb]{0,0,0}$b$}%
}}}}
\put(2713, 42){\makebox(0,0)[b]{\smash{{\SetFigFont{5}{6.0}{\rmdefault}{\mddefault}{\updefault}{\color[rgb]{0,0,0}$c$}%
}}}}
\put(2282, -9){\makebox(0,0)[b]{\smash{{\SetFigFont{5}{6.0}{\rmdefault}{\mddefault}{\updefault}{\color[rgb]{0,0,0}$2$}%
}}}}
\put(1201,291){\makebox(0,0)[b]{\smash{{\SetFigFont{5}{6.0}{\rmdefault}{\mddefault}{\updefault}{\color[rgb]{0,0,0}$c$}%
}}}}
\put(1215,-600){\makebox(0,0)[b]{\smash{{\SetFigFont{5}{6.0}{\rmdefault}{\mddefault}{\updefault}{\color[rgb]{0,0,0}$3$}%
}}}}
\put(826,-136){\makebox(0,0)[b]{\smash{{\SetFigFont{5}{6.0}{\rmdefault}{\mddefault}{\updefault}{\color[rgb]{0,0,0}$a$}%
}}}}
\put(1501,-211){\makebox(0,0)[b]{\smash{{\SetFigFont{5}{6.0}{\rmdefault}{\mddefault}{\updefault}{\color[rgb]{0,0,0}$a$}%
}}}}
\put(1583, 17){\makebox(0,0)[b]{\smash{{\SetFigFont{5}{6.0}{\rmdefault}{\mddefault}{\updefault}{\color[rgb]{0,0,0}$2$}%
}}}}
\put(4767,974){\makebox(0,0)[b]{\smash{{\SetFigFont{5}{6.0}{\rmdefault}{\mddefault}{\updefault}{\color[rgb]{0,0,0}$1$}%
}}}}
\put(5395,349){\makebox(0,0)[b]{\smash{{\SetFigFont{5}{6.0}{\rmdefault}{\mddefault}{\updefault}{\color[rgb]{0,0,0}$4$}%
}}}}
\put(4622,-192){\makebox(0,0)[b]{\smash{{\SetFigFont{5}{6.0}{\rmdefault}{\mddefault}{\updefault}{\color[rgb]{0,0,0}$a$}%
}}}}
\put(4618,506){\makebox(0,0)[b]{\smash{{\SetFigFont{5}{6.0}{\rmdefault}{\mddefault}{\updefault}{\color[rgb]{0,0,0}$b$}%
}}}}
\put(5066, 60){\makebox(0,0)[b]{\smash{{\SetFigFont{5}{6.0}{\rmdefault}{\mddefault}{\updefault}{\color[rgb]{0,0,0}$c$}%
}}}}
\put(4433,292){\makebox(0,0)[b]{\smash{{\SetFigFont{5}{6.0}{\rmdefault}{\mddefault}{\updefault}{\color[rgb]{0,0,0}$c$}%
}}}}
\put(4447,-599){\makebox(0,0)[b]{\smash{{\SetFigFont{5}{6.0}{\rmdefault}{\mddefault}{\updefault}{\color[rgb]{0,0,0}$3$}%
}}}}
\put(4058,-135){\makebox(0,0)[b]{\smash{{\SetFigFont{5}{6.0}{\rmdefault}{\mddefault}{\updefault}{\color[rgb]{0,0,0}$a$}%
}}}}
\put(4932,331){\makebox(0,0)[b]{\smash{{\SetFigFont{5}{6.0}{\rmdefault}{\mddefault}{\updefault}{\color[rgb]{0,0,0}$2$}%
}}}}
\put(4964,-282){\makebox(0,0)[b]{\smash{{\SetFigFont{5}{6.0}{\rmdefault}{\mddefault}{\updefault}{\color[rgb]{0,0,0}$a$}%
}}}}
\put(7201,-661){\makebox(0,0)[b]{\smash{{\SetFigFont{6}{7.2}{\rmdefault}{\mddefault}{\updefault}{\color[rgb]{0,0,0}(b)}%
}}}}
\put(3602,2247){\makebox(0,0)[b]{\smash{{\SetFigFont{5}{6.0}{\rmdefault}{\mddefault}{\updefault}{\color[rgb]{0,0,0}$\longrightarrow$}%
}}}}
\put(1956,2255){\makebox(0,0)[b]{\smash{{\SetFigFont{8}{9.6}{\rmdefault}{\mddefault}{\updefault}{\color[rgb]{0,0,0}$\vrpop{}{}$}%
}}}}
\put(1208,3055){\makebox(0,0)[b]{\smash{{\SetFigFont{5}{6.0}{\rmdefault}{\mddefault}{\updefault}{\color[rgb]{0,0,0}$1$}%
}}}}
\put(1490,2658){\makebox(0,0)[b]{\smash{{\SetFigFont{5}{6.0}{\rmdefault}{\mddefault}{\updefault}{\color[rgb]{0,0,0}$b$}%
}}}}
\put(2414,3056){\makebox(0,0)[b]{\smash{{\SetFigFont{5}{6.0}{\rmdefault}{\mddefault}{\updefault}{\color[rgb]{0,0,0}$2$}%
}}}}
\put(2269,1890){\makebox(0,0)[b]{\smash{{\SetFigFont{5}{6.0}{\rmdefault}{\mddefault}{\updefault}{\color[rgb]{0,0,0}$a$}%
}}}}
\put(2265,2588){\makebox(0,0)[b]{\smash{{\SetFigFont{5}{6.0}{\rmdefault}{\mddefault}{\updefault}{\color[rgb]{0,0,0}$b$}%
}}}}
\put(2713,2142){\makebox(0,0)[b]{\smash{{\SetFigFont{5}{6.0}{\rmdefault}{\mddefault}{\updefault}{\color[rgb]{0,0,0}$c$}%
}}}}
\put(2282,2091){\makebox(0,0)[b]{\smash{{\SetFigFont{5}{6.0}{\rmdefault}{\mddefault}{\updefault}{\color[rgb]{0,0,0}$2$}%
}}}}
\put(1201,2391){\makebox(0,0)[b]{\smash{{\SetFigFont{5}{6.0}{\rmdefault}{\mddefault}{\updefault}{\color[rgb]{0,0,0}$c$}%
}}}}
\put(1215,1500){\makebox(0,0)[b]{\smash{{\SetFigFont{5}{6.0}{\rmdefault}{\mddefault}{\updefault}{\color[rgb]{0,0,0}$1$}%
}}}}
\put(826,1964){\makebox(0,0)[b]{\smash{{\SetFigFont{5}{6.0}{\rmdefault}{\mddefault}{\updefault}{\color[rgb]{0,0,0}$a$}%
}}}}
\put(1501,1889){\makebox(0,0)[b]{\smash{{\SetFigFont{5}{6.0}{\rmdefault}{\mddefault}{\updefault}{\color[rgb]{0,0,0}$a$}%
}}}}
\put(1583,2117){\makebox(0,0)[b]{\smash{{\SetFigFont{5}{6.0}{\rmdefault}{\mddefault}{\updefault}{\color[rgb]{0,0,0}$1$}%
}}}}
\put(4433,3055){\makebox(0,0)[b]{\smash{{\SetFigFont{5}{6.0}{\rmdefault}{\mddefault}{\updefault}{\color[rgb]{0,0,0}$1$}%
}}}}
\put(4715,2658){\makebox(0,0)[b]{\smash{{\SetFigFont{5}{6.0}{\rmdefault}{\mddefault}{\updefault}{\color[rgb]{0,0,0}$b$}%
}}}}
\put(5639,3056){\makebox(0,0)[b]{\smash{{\SetFigFont{5}{6.0}{\rmdefault}{\mddefault}{\updefault}{\color[rgb]{0,0,0}$2$}%
}}}}
\put(5494,1890){\makebox(0,0)[b]{\smash{{\SetFigFont{5}{6.0}{\rmdefault}{\mddefault}{\updefault}{\color[rgb]{0,0,0}$a$}%
}}}}
\put(5490,2588){\makebox(0,0)[b]{\smash{{\SetFigFont{5}{6.0}{\rmdefault}{\mddefault}{\updefault}{\color[rgb]{0,0,0}$b$}%
}}}}
\put(5938,2142){\makebox(0,0)[b]{\smash{{\SetFigFont{5}{6.0}{\rmdefault}{\mddefault}{\updefault}{\color[rgb]{0,0,0}$c$}%
}}}}
\put(5507,2091){\makebox(0,0)[b]{\smash{{\SetFigFont{5}{6.0}{\rmdefault}{\mddefault}{\updefault}{\color[rgb]{0,0,0}$2$}%
}}}}
\put(4426,2391){\makebox(0,0)[b]{\smash{{\SetFigFont{5}{6.0}{\rmdefault}{\mddefault}{\updefault}{\color[rgb]{0,0,0}$c$}%
}}}}
\put(4440,1500){\makebox(0,0)[b]{\smash{{\SetFigFont{5}{6.0}{\rmdefault}{\mddefault}{\updefault}{\color[rgb]{0,0,0}$1$}%
}}}}
\put(4051,1964){\makebox(0,0)[b]{\smash{{\SetFigFont{5}{6.0}{\rmdefault}{\mddefault}{\updefault}{\color[rgb]{0,0,0}$a$}%
}}}}
\put(4726,1889){\makebox(0,0)[b]{\smash{{\SetFigFont{5}{6.0}{\rmdefault}{\mddefault}{\updefault}{\color[rgb]{0,0,0}$a$}%
}}}}
\put(4808,2117){\makebox(0,0)[b]{\smash{{\SetFigFont{5}{6.0}{\rmdefault}{\mddefault}{\updefault}{\color[rgb]{0,0,0}$1$}%
}}}}
\put(7126,2264){\makebox(0,0)[b]{\smash{{\SetFigFont{5}{6.0}{\rmdefault}{\mddefault}{\updefault}{\color[rgb]{0,0,0}$\arrow{\addedge{\elab}{1}{2}}{}$}%
}}}}
\put(8108,3055){\makebox(0,0)[b]{\smash{{\SetFigFont{5}{6.0}{\rmdefault}{\mddefault}{\updefault}{\color[rgb]{0,0,0}$1$}%
}}}}
\put(9314,3056){\makebox(0,0)[b]{\smash{{\SetFigFont{5}{6.0}{\rmdefault}{\mddefault}{\updefault}{\color[rgb]{0,0,0}$2$}%
}}}}
\put(9182,2091){\makebox(0,0)[b]{\smash{{\SetFigFont{5}{6.0}{\rmdefault}{\mddefault}{\updefault}{\color[rgb]{0,0,0}$2$}%
}}}}
\put(8101,2391){\makebox(0,0)[b]{\smash{{\SetFigFont{5}{6.0}{\rmdefault}{\mddefault}{\updefault}{\color[rgb]{0,0,0}$c$}%
}}}}
\put(8115,1500){\makebox(0,0)[b]{\smash{{\SetFigFont{5}{6.0}{\rmdefault}{\mddefault}{\updefault}{\color[rgb]{0,0,0}$1$}%
}}}}
\put(7726,1964){\makebox(0,0)[b]{\smash{{\SetFigFont{5}{6.0}{\rmdefault}{\mddefault}{\updefault}{\color[rgb]{0,0,0}$a$}%
}}}}
\put(8483,2117){\makebox(0,0)[b]{\smash{{\SetFigFont{5}{6.0}{\rmdefault}{\mddefault}{\updefault}{\color[rgb]{0,0,0}$1$}%
}}}}
\put(8689,3003){\makebox(0,0)[b]{\smash{{\SetFigFont{5}{6.0}{\rmdefault}{\mddefault}{\updefault}{\color[rgb]{1,0,0}$\elab$}%
}}}}
\put(8120,1957){\makebox(0,0)[b]{\smash{{\SetFigFont{5}{6.0}{\rmdefault}{\mddefault}{\updefault}{\color[rgb]{0,0,0}$a$}%
}}}}
\put(8075,2658){\makebox(0,0)[b]{\smash{{\SetFigFont{5}{6.0}{\rmdefault}{\mddefault}{\updefault}{\color[rgb]{0,0,0}$b$}%
}}}}
\put(2401,1514){\makebox(0,0)[b]{\smash{{\SetFigFont{5}{6.0}{\rmdefault}{\mddefault}{\updefault}{\color[rgb]{0,0,0}$2$}%
}}}}
\put(5626,1514){\makebox(0,0)[b]{\smash{{\SetFigFont{5}{6.0}{\rmdefault}{\mddefault}{\updefault}{\color[rgb]{0,0,0}$2$}%
}}}}
\put(9301,1514){\makebox(0,0)[b]{\smash{{\SetFigFont{5}{6.0}{\rmdefault}{\mddefault}{\updefault}{\color[rgb]{0,0,0}$2$}%
}}}}
\put(9611,2142){\makebox(0,0)[b]{\smash{{\SetFigFont{5}{6.0}{\rmdefault}{\mddefault}{\updefault}{\color[rgb]{0,0,0}$c$}%
}}}}
\put(9427,2567){\makebox(0,0)[b]{\smash{{\SetFigFont{5}{6.0}{\rmdefault}{\mddefault}{\updefault}{\color[rgb]{0,0,0}$b$}%
}}}}
\put(9429,1967){\makebox(0,0)[b]{\smash{{\SetFigFont{5}{6.0}{\rmdefault}{\mddefault}{\updefault}{\color[rgb]{0,0,0}$a$}%
}}}}
\end{picture}%

%% file: figure-proctype.tex
\begin{figure}[t!]
  \vspace*{-\baselineskip}
  \begin{center}
    \begin{minipage}{.45\textwidth}
      \scalebox{0.7}{
        \begin{tikzpicture}[node distance=1.5cm]
          \tikzstyle{every state}=[inner sep=3pt,minimum size=20pt]
          \tikzset{
            arr/.style={->,line width=1pt,thick}
          }
          \node[petri-p,draw=black,label=-90:{on}] (on) at (0,0) {};
          \node[petri-p,draw=black,label=90:{off}] (off) at ($(on) + (0,2)$) {};
          \node[petri-thor,draw=black,fill=black,label=0:{send}] (send) at ($(on) + (0.7,1)$) {};
          \draw[arr] (on) edge[out=0,in=-90] (send);
          \draw[arr] (send) edge[out=90,in=0] (off);
          \node[petri-tok] at (on) {};

          \node at ($(on) + (0,-1)$) {Process type \emph{Once}};

          \node[petri-p,draw=black,label=-90:{free}] (free) at (3.5,0) {};
          \node[petri-p,draw=black,label=90:{busy}] (busy) at ($(free) + (0,2)$) {};
          \node[petri-thor,draw=black,fill=black,label=180:{recv}] (recv) at ($(free) + (-0.7,1)$) {};
          \node[petri-thor,draw=black,fill=yellow,label=0:{handle}] (handle) at ($(free) + (0.7,1)$) {};
          \draw[arr] (free) edge[out=180,in=-90] (recv);
          \draw[arr] (recv) edge[out=90,in=180] (busy);
          \draw[arr] (busy) edge[out=0,in=90] (handle);
          \draw[arr] (handle) edge[out=-90,in=0] (free);
          \node[petri-tok] at (free) {};

          \node at ($(free) + (0,-1)$) {Process type \emph{Loop}};
        \end{tikzpicture}}
      \centerline{(a)}
    \end{minipage}
    \begin{minipage}{.5\textwidth}
      \scalebox{0.6}{
        \begin{tikzpicture}
          \tikzset{
            arr/.style={->,line width=1pt,thick}
          }
          \foreach \i in {0,...,3} {
            \node[petri-p,label=-90:{on}] (on\i) at (2.8*\i,0) {};
            \node[petri-p,label=-90:{off}] (off\i) at ($(on\i) + (1,0)$) {};
            \node[petri-tok] at (on\i) {};
            
          }
          \foreach \j in {0,...,2} {
            \node[petri-p,label=120:{free}] (free\j) at (4.2*\j,3.2) {};
            \node[petri-p,label=60:{busy}] (busy\j) at ($(free\j) + (1,0)$) {};
            \node[petri-tok] at (free\j) {};
            \node[petri-tver,fill=yellow] (int\j) at ($(free\j) + (0.5,0.5)$) {};
            \draw[arr] (busy\j) edge[bend right=40] (int\j);
            \draw[arr] (int\j) edge[bend right=40] (free\j);
          }
          \foreach \i in {0,...,3} {
            \foreach \j in {0,...,2} {
              \node[petri-tver,fill=black] (obs\i\j) at (-0.5+2.8*\i+\j,1) {};
            }
          }
          \foreach \i in {0,...,3} {
            \foreach \j in {0,...,2} {
              \draw[arr] (on\i) edge[bend left=20] (obs\i\j);
              \draw[arr] (obs\i\j) edge[bend left=20] (off\i);
              \draw[arr] (free\j) edge[bend right=5] (obs\i\j);
              \draw[arr] (obs\i\j) edge[bend right=5] (busy\j);
            }
          }
      \end{tikzpicture}}
      
      \centerline{(b)}
    \end{minipage}
  \end{center}
  \vspace*{-1.9\baselineskip}
\caption{Examples of process types (a). The behavior of a $\protect\overrightarrow{K}_{4,3}(\mathrm{send},\mathrm{recv})$ system (b)}
\label{fig:proctype}
\vspace*{-\baselineskip}
\end{figure}

%% file: translation.tex
\section{Translating \vrtext{} to \hrtext{} Systems}

The existence of several verification techniques for parametric
systems with \hrtext{} architectures and the scarsity of similar
results for \vrtext{} systems motivates us to define a translation
from \vrtext{} to \hrtext{} systems, that preserves the answer of
the $\prp{}{}{}{}$ decision problem (\defref{def:pmcp}).
This translation uses expansion to encode dense graphs as sparse graphs
(see \figref{fig:K43}).
Since labeled graphs are used as system models (\defref{def:system}),
the behavior of a sparse graph must be equivalent to the behavior of its expansion,
on what concerns the satisfiability of reachability formul{\ae},
thus reducing each instance of the $\prp{}{}{}{}$ problem for a \vrtext{}
grammar to an instance of the same problem for an \hrtext{} grammar.

\begin{theorem}\label{thm:main}
  There exist computable $\ptypes^\expand$ and $\syslab^\expand$,
  process types and place labeling respectively,
  such that for any \vrtext-grammar $\grammar$,
  variable labeling $\syslab : \placeof{\ptypes} \rightarrow \vars$
  and reachability formula $\alpha$,
  one can effectively build a \hrtext-grammar $\grammar'$
  such that
  \(\prp{\grammar}{\alpha}{\syslab}{\ptypes} \iff
  \prp{\grammar'}{\alpha}{\syslab^\expand}{\ptypes^\expand}\)
\end{theorem}
\noindent This result shows that it is possible to solve each instance
of the parametric model checking problem that takes a \vrtext{}
grammar as input by an effective reduction to an instance of the same
problem for a \hrtext{} grammar, which, as previously mentioned, has
received more attention lately~\cite{arXiv2025}.

Intuitively the translation works as follows:
vertices in the \hrtext{} system are either
``real'' vertices from the original \vrtext{} system, or routing nodes.
A unique vertex of the \hrtext{} system will act as representative of
all vertices currently $\plab$-ports in the \vrtext{} system,
and communication between a $\plab$-port $v$ and a $\plab'$-port $v'$
in the \vrtext{} system will in the \hrtext{} system roughly take the form of
(1) a request from $v$ to its representative $v_\plab$ through a chain of routing nodes,
(2) similarly a request from $v'$ to its representative $v_{\plab'}$,
(3) a direct communication between $v_\plab$ and $v_{\plab'}$
(4) an acknowledgement from $v_\plab$ to $v$ following the inverse chain of routing nodes,
(5) an acknowledgement from $v_{\plab'}$ to $v'$.
This is a simplification because in reality
steps (1--2) and (4--5) above may be interleaved,
and rather than a single representative there is
one representative per observable transition.
We define below the process types that correspond to routing nodes,
give an intuition of how they function,
then show the actual inductive translation.

\input{figure-halfprocs}

Let $\ptypes$ be the set of process types used by the \vrtext{}
grammars, in the rest of this section. We define a new set of process
types $\ptypes^\expand$ such that $\placeof{\ptypes} \subseteq
\placeof{\ptypes^\expand}$, see \figref{fig:halfprocs} (b) for an
example. First, let $\halve{\ptypes} \isdef
\set{\halve{\ptype} \mid \ptype \in \ptypes}$ be a set of process
types disjoint from $\ptypes$.
A generic $\halve{\ptype}$ is defined below from $\ptype \in \ptypes$.
Intuitively, $\halve{\ptype}$ simulates the behavior of $\ptype$
in the context of communication through routing nodes:
each observable transition is split into two halves,
one sending a request, then eventually receiving the acknowledgement.
\begin{compactitem}[-]
\item $\placeof{\halve{\ptype}} \isdef \placeof{\ptype} \uplus
  \set{\halfplace{t} \mid t \in \obstransof{\ptype}}$, each
  $\halfplace{t}$ is a new place associated
  with the transition $t$,
\item $\transof{\halve{\ptype}} \isdef \inttransof{\ptype} \uplus
  \set{\gttry{t}, \gtcommit{t} \mid t \in \obstransof{\ptype}}$, \ie
  each observable transition $t$ of $\ptype$ is split into
  $\gttry{t}$ (an attempt at firing $t$) and $\gtcommit{t}$
 ($t$ has been fired remotely),
\item $\weightof{\halve{\ptype}} \isdef
  \proj{\weightof{\ptype}}{\placeof{\ptype} \times \inttransof{\ptype}
    \uplus \inttransof{\ptype} \times \placeof{\ptype}} \uplus
  \set{(q,\gttry{t},1),(\gttry{t},\halfplace{t},1) \mid q \in
    \placeof{\ptype},~ t \in \obstransof{\ptype},~
    \weightof{\ptype}(q,t)=1} \uplus
  \set{(\halfplace{t},\gtcommit{t},1),(\gtcommit{t},q,1) \mid t \in
    \obstransof{\ptype},~ \weightof{\ptype}(t,q)=1}$, \ie whenever $t$
  is an observable transition from $q$ to $q'$, there are observable
  transitions $\gttry{t}$ from $q$ to $\halfplace{t}$ and
  $\gtcommit{t}$ from $\halfplace{t}$ to $q'$ in $\halve{\ptype}$.
\item $\initmarkof{\halve{\ptype}} \isdef \initmarkof{\ptype}$. 
\end{compactitem}
Second, for each observable transition $t \in \obstransof{\ptypes}$,
we consider a new process type $\epsilon_t$, defined below.
See \figref{fig:halfprocs} (c) for an example.
This is what we call the ``routing nodes''.
\begin{compactitem}[-]
\item $\placeof{\epsilon_t} \isdef \set{\gtidle_t, \gtactive_t, \gtwait_t,
  \gtreply_t}$,
\item $\transof{\epsilon_t} = \obstransof{\epsilon_t} \isdef
  \set{\gtrecv, \gtfwd, \gtack, \gtreset, t}$,
\item $\weightof{\epsilon_t}$ consists of the following edges, all of weight $1$:
  \begin{align*}
    \prepost{\gtrecv} \isdef & (\gtidle_t, \gtactive_t) \hspace*{4mm} \prepost{\gtfwd} \isdef (\gtactive_t, \gtwait_t)
    \hspace*{4mm} \prepost{t} \isdef (\gtactive_t, \gtreply_t) \\
    \prepost{\gtack}  \isdef & (\gtwait_t, \gtreply_t) \hspace*{4mm} \prepost{\gtreset} \isdef (\gtreply_t, \gtidle_t)
  \end{align*}
\item $\initmarkof{\epsilon_t}(\gtidle_t) \isdef 1$ and
  $\initmarkof{\epsilon_t}(\gtactive_t) =
  \initmarkof{\epsilon_t}(\gtwait_t) = \initmarkof{\epsilon_t}(\gtreply_t)
  \isdef 0$.
\end{compactitem}
Finally, $\ptypes^\expand \isdef \halve{\ptypes} \cup \set{\epsilon_t
  \mid t \in \obstransof{\ptype}}$,
  ``real vertices'' and ``routing nodes'' respectively.
\begin{example}\label{ex:halves}
  \figref{fig:halfprocs} (a) shows examples of $\ptype$, $\halve{\ptype}$,
  and $\epsilon_t$.
  Instances of $\epsilon_t$ exist for
  $t \in \obstransof{\ptype} = \set{\mathsf{up},\mathsf{dn}}$.
  Two processes of type $\ptype$ communicate via a
  $(\mathsf{up},\mathsf{dn})$ transition, which is encoded by a
  sequence of transitions, see \figref{fig:routing-gadget}:
  {\small\[\begin{array}{l}
  (\underbrace{\gttry{\mathsf{up}}}_{\halve{\ptype}}, \underbrace{\gtrecv),
    (\gtfwd,\gtrecv), \ldots, (\gtfwd,\gtrecv),
    (\mathsf{up}}_{\epsilon_{\mathsf{up}}},
  \underbrace{\mathsf{dn}}_{\epsilon_{\mathsf{dn}}}), 
  \underbrace{(\gtreset,\gtack), \ldots,
    (\gtreset,\gtack), 
    (\gtreset}_{\epsilon_{\mathsf{up}}},
  \underbrace{\gtcommit{\mathsf{up}}}_{\halve{\ptype}})
  \end{array}\]}
  This matches the earlier intuition:
  the leftmost process tries to execute $\mathsf{up}$
  and sends a request through routers $\epsilon_\mathsf{up}$.
  At the same time, the right process tries to execute
  $\mathsf{dn}$ and sends a request through routers $\epsilon_\mathsf{dn}$.
  When both requests reach adjacent routers,
  the transition $(\mathsf{up},\mathsf{dn})$ is jointly executed.
  Following that, an acknowledgement is sent back to the processes that initiated
  the communication as a sequence of $(\gtreset,\gtack)$.
\end{example}
\input{figure-router}

The variable labeling $\syslab : \placeof{\ptypes} \rightharpoonup
\vars$ induces the following labeling $\syslab^\expand :
\placeof{\ptypes^\expand} \rightharpoonup \vars$: \begin{align*}
  \syslab^\expand(q) \isdef & ~\syslab(q) \text{, for each } q \in \placeof{\ptypes} \\
  \syslab^\expand(\gtactive_t) \isdef & ~\syslab(\pre{t}) \hspace*{8mm}
  \syslab^\expand(\gtreply_t) \isdef \syslab(\post{t})\text{, for each } t \in \obstransof{\ptypes}  \\
  & \text{undefined everywhere else}
\end{align*}
The rest of the paper is the proof of \thmref{thm:main},
instanciated with these $\ptypes^\expand$ and $\syslab^\expand$.

\subsection{Proof of \thmref{thm:main}}

The main idea is to use the $\epsilon_t$ processes to define, for a
\vrtext-term $\theta$, a \hrtext-term $\expand(\theta)$ that evaluates
to a sparse system, from which the system $\val{\vr}{\theta}$ can be
retrived using graph expansion, \ie $\val{\vr}{\theta} =
\expof{\val{\hr}{\expand(\theta)}}$ (\defref{def:expansion}).
Using the set of port labels
\(
\palpha^\expand \isdef \set{\halve{\plab}, (\plab,t),
  \overline{(\plab,t)} \mid \plab\in\palpha,
  t\in\obstransof{\ptypeof{}(\plab)}}
\)
where the process type of $\halve{\plab}$ is
$\halve{(\ptypeof{}(\plab))}$ and the process type of $(\plab,t)$ and
$\overline{(\plab,t)}$ is $\epsilon_t$,
we define $\expand(\theta)$ inductively on the structure of $\theta$.
That is, in addition to a
source label $\halve{\plab}$ for each port label $\plab \in \palpha$
that occurs in $\theta$, we add two copies of each port label for each
observable transition.  Intuitively, $(\plab,t)$ will be the port
label denoting the root of the tree of $\epsilon_t$ processes, and
$\overline{(\plab,t)}$ will be used as backup when we need a fresh
copy of $(\plab,t)$. We define a renaming function $\mathsf{fresh} :
\palpha^\expand \partial \palpha^\expand$ by
$\mathsf{fresh}(\overline{(\plab,t)}) = (\plab,t)$, and undefined
everywhere else. That way $\mathsf{fresh}$ substitutes each source for
its copy and forgets the original one.
We also use the renaming function
$\mathsf{obs} : \palpha^\expand \partial \palpha^\expand$
which is the identity on $\palpha\times\obstransof{\ptypes}$
and undefined everywhere else,
\ie $\mathsf{obs}$ forgets every port label except the $(\plab,t)$ labels.

Concretely, creating new edges between $\epsilon_t$-vertices is done by the
below functions:
\vspace*{-0.5\baselineskip}
\[\begin{array}{rll}
E^t_\epsilon(\plab,\plab') \isdef &~ \edge{(\gtfwd,\gtrecv)}{(\plab,t)}{\overline{(\plab',t)}} \hrpop{}{} \edge{(\gtreset,\gtack)}{\overline{(\plab',t)}}{(\plab,t)}
&
\raisebox{-5mm}{\scalebox{0.6}{
  \begin{tikzpicture}
    \tikzset{
      proc/.style={circle,fill=black,inner sep=0pt,minimum width=3pt},
      arr/.style={->,line width=0.3pt},
    }
    \node[proc,label=-180:{$(\plab,t)$}] (p1) at (0,0) {};
    \node[proc,label=0:{$\overline{(\plab',t)}$}] (p2) at (2,0) {};
    \draw[arr] (p1) to[bend left=20] node[above]{$(\gtfwd,\gtrecv)$} (p2);
    \draw[arr] (p2) to[bend left=20] node[below]{$(\gtreset,\gtack)$} (p1);
  \end{tikzpicture}
}}
\\
\halve{E}^t(\plab) \isdef &~ \edge{(\gttry{t},\gtrecv)}{\halve{\plab}}{(\plab,t)} \hrpop{}{} \edge{(\gtreset,\gtcommit{t})}{(\plab,t)}{\halve{\plab}}
&
\raisebox{-5mm}{\scalebox{0.6}{
  \begin{tikzpicture}
    \tikzset{
      proc/.style={circle,fill=black,inner sep=0pt,minimum width=3pt},
      arr/.style={->,line width=0.3pt},
    }
    \node[proc,label=-180:{$\halve{\plab}$}] (p1) at (0,0) {};
    \node[proc,label=0:{$(\plab,t)$}] (p2) at (2,0) {};
    \draw[arr] (p1) to[bend left=20] node[above]{$(\gttry{t},\gtrecv)$} (p2);
    \draw[arr] (p2) to[bend left=20] node[below]{$(\gtreset,\gtcommit{t})$} (p1);
  \end{tikzpicture}
}}
\end{array}\]
Recall that the main invariant during the translation is that
whenever $v$ in the \vrtext{} system is a $\plab$-port,
its corresponding vertex in the \hrtext{} system is linked to the unique
representative of $(\plab,t)$ through a chain of routers $\epsilon_t$.

An essential gadget used by the construction is the encoding of a
relabeling function $\alpha : \palpha \partial \palpha$
(which is by definition type-preserving) by a graph
denoted $\mathsf{enc}(\alpha)$.  We show an example of
$\mathsf{enc}(\alpha)$, where $\alpha$ is the function $[\plab_1
  \mapsto \plab_2, \plab_2 \mapsto \plab_1, \plab_3 \mapsto \plab_4]$.
So as to not clutter the figure, we simply write $E_\epsilon$ to label
the entire bundle of edges $\set{E^t_\epsilon \mid t\in
  \obstransof{\ptypes}}$ and we write $(\plab,\_)$ for $\set{(\plab,t)
  \mid t \in \obstransof{\ptypeof{}(\plab)}}$.

\vspace*{-0.5\baselineskip}
\begin{minipage}{0.5\textwidth}
  \begin{align*}
    \mathsf{enc}(\alpha) \isdef
    \bighrpop{\begin{array}{c}
        \scriptscriptstyle{\plab,\plab' \in \palpha} \\[-2mm]
        \scriptscriptstyle{\alpha(\plab) = \plab'}
    \end{array}}
    \bighrpop{\begin{array}{c}
        \scriptscriptstyle{t \in \obstransof{\ptypeof{}(\plab)}}
    \end{array}}
    E^t_\epsilon(\plab,\plab')
  \end{align*}
\end{minipage}
\begin{minipage}{0.5\textwidth}
  \begin{center}
  \scalebox{0.7}{
  \begin{tikzpicture}
    \tikzset{
      proc/.style={circle,fill=black,inner sep=0pt,minimum width=3pt},
      arr/.style={->,line width=0.3pt},
      brace/.style={decorate,decoration={brace,mirror,raise=3pt}},
    }
    \pgfmathsetmacro\step{0.9}
    \pgfmathsetmacro\mini{0.11}
    \foreach \pid/\pos/\label in {1/1/1,2/2/2,3/3/3} {
      \node[proc] (p\pid-t1l) at (0,\pos*\step + 1*\mini) {};
      \node[proc] (p\pid-t2l) at (0,\pos*\step + 2*\mini) {};
      \node[proc] (p\pid-t3l) at (0,\pos*\step + 3*\mini) {};
      \node[proc] (p\pid-t4l) at (0,\pos*\step + 4*\mini) {};
      \draw[brace] ($(p\pid-t4l) + (0,0.5*\mini)$) to node[left=5pt]{$(\plab_\label,\_)$} ($(p\pid-t1l) + (0,-0.5*\mini)$);
    }

    \foreach \pid/\pos/\prev/\label in {1/1/2/1,2/2/1/2,3/3/3/4} {
      \node[proc] (p\pid-t1r) at (3,\pos*\step + 1*\mini) {};
      \node[proc] (p\pid-t2r) at (3,\pos*\step + 2*\mini) {};
      \node[proc] (p\pid-t3r) at (3,\pos*\step + 3*\mini) {};
      \node[proc] (p\pid-t4r) at (3,\pos*\step + 4*\mini) {};
      \draw[brace] ($(p\pid-t1r) + (0,-0.5*\mini)$) to node[right=5pt]{$\overline{(\plab_\label,\_)}$} ($(p\pid-t4r) + (0,0.5*\mini)$);

      \draw[arr] (p\prev-t1l) to (p\pid-t1r);
      \draw[arr] (p\prev-t2l) to (p\pid-t2r);
      \draw[arr] (p\prev-t3l) to (p\pid-t3r);
      \draw[arr] (p\prev-t4l) to node[above,pos=0.95]{$E_\epsilon$}
                                  (p\pid-t4r);
    }
  \end{tikzpicture}
  }
  \end{center}
\end{minipage}
Intuitively, as the renaming $\alpha$
turns every $\plab$-port into an $\alpha(\plab)$-port,
$\mathsf{enc}(\alpha)$ links the previous representatives of $(\plab,\_)$
to the new representatives of $(\alpha(\plab),\_)$, preserving the invariant.
The translation $\expand$ is defined inductively on
the structure of the term $\theta$:
\[\begin{array}{rll}
  \expand(\svertex{\plab}{}) \isdef & ~\relab{\mathsf{obs}}
  \Big(\svertex{\halve{\plab}} ~\hrpop{}{}~ \Big(\bighrpop{
    \begin{array}{c}
      \scriptscriptstyle{t\in\obstransof{\ptypeof{}(\plab)}} \\[-2mm]
    \end{array}
  }
    \halve{E}^t(\plab)\Big)\Big)
  & \text{ \figref{fig:h} (a)}
  \\
  \expand(\addedge{(t,t')}{\plab}{\plab'}{}(\theta_1)) \isdef &
  ~\edge{(t,t')}{(\plab,t)}{(\plab',t')}{} \hrpop{}{} \expand(\theta_1)
  & \text{ \figref{fig:h} (b)}
  \\
  \expand(\relab{\alpha}{}(\theta_1)) \isdef & ~\relab{\mathsf{fresh}}{}\Big(
  \expand(\theta_1) \hrpop{}{} \mathsf{enc}(\alpha)\Big)
  & \text{ \figref{fig:h} (c,c',c'')}
  \\
  \expand(\theta_1 \vrpop{\plabs_1}{\plabs_2} \theta_2) \isdef &
     \relab{\mathsf{fresh}}{}\left(
       \expand(\theta_1)
       \hrpop{}{}
       \mathsf{enc}(\fnid_{\plabs_1\cap\plabs_2})
     \right) \\
     & \hrpop{}{}
       \relab{\mathsf{fresh}}{}\left(
       \expand(\theta_2)
       \hrpop{}{}
       \mathsf{enc}(\fnid_{\plabs_1\cap\plabs_2})
     \right)
  & \text{ \figref{fig:h} (d,d',d'')}
\end{array}\]
\input{figure-h}
These are quite straightforward once we understand the invariant:
the translation of $\svertex{\plab}$ creates a new vertex
and links it to all relevant routing nodes,
the rule for $\addedge{(t,t')}{\plab}{\plab'}$ establishes a communication
between the representatives of $(\plab,t)$ and $(\plab',t')$,
the one for $\vrpop{}{}$ creates a unique common representative,
and the relabeling is as explained earlier.

The relation between the systems $\val{\vr}{\theta}$ and
$\val{\hr}{\expand(\theta)}$ is captured by the below lemma. Its proof
is the same as in~\cite[Proposition 2.4]{COURCELLE1995275}, where the
new edges have the same label, instead of the following sets of edge
labels (\defref{def:expansion}): \begin{align*} \expalpha \isdef &
  ~\set{(\gttry{t},\gtrecv) \mid t \in \obstransof{\ptypes}} \cup
  \set{(\gtfwd,\gtrecv)} \\ \overleftarrow{\expalpha} \isdef &
  ~\set{(\gtreset,\gtcommit{t}) \mid t \in \obstransof{\ptypes}} \cup
  \set{(\gtreset,\gtack)}
\end{align*}
Here we denote $(\gtcommit{t},\gtreset) \isdef
\overleftarrow{(\gttry{t},\gtrecv)}$ and
$(\gtack,\gtreset)\isdef\overleftarrow{(\gtfwd,\gtrecv)}$
(\defref{def:expansion}). The meaning and orientation of these edges
are shown on an example in \figref{fig:routing-gadget}.

\begin{lemma}\label{lemma:expansion}
  Let $\theta$ be a ground \vrtext-term, 
  $\asys \isdef \val{\vr}{\theta}$ and $\asys' \isdef \val{\hr}{\expand(\theta)}$.
  Then, $\asys = \expof{\asys'}$. For each $v,v' \in \vertof{\asys'}$,
  if $E \neq \emptyset$ is the set of edges from $v$ to $v'$, then exactly one of the following holds: \begin{compactenum}
  \item $(\procof{\asys'}(v), \procof{\asys'}(v')) =
    (\halve{\ptype}, \epsilon_t)$,
    $E = \set{(\gttry{t},\gtrecv), (\gtcommit{t},\gtreset)}$ and
    $t \in \obstransof{\ptype}$;
  \item $(\procof{\asys'}(v), \procof{\asys'}(v')) = (\epsilon_t,\epsilon_{t'})$
    and $E = \set{(t,t')}$;
  \item $\procof{\asys'}(v) = \procof{\asys'}(v') = \epsilon_t$
    and $E = \set{(\gtfwd,\gtrecv), (\gtack,\gtreset)}$.
    %
  \end{compactenum}
\end{lemma}
\begin{proof}
  The proof of $\asys = \expof{\asys'}$ is the same as \cite[Proposition 2.4]{COURCELLE1995275}.
  The rest of the points follow from the definition of $\expand$. 
  \qed
\end{proof}

The behaviors of $\val{\vr}{\theta}$ and $\val{\hr}{\expand(\theta)}$
are related by the following notion of path stuttering 
(inspired by~\cite[Definition 7.89]{DBLP:books/daglib/0020348}):

\begin{definition}\label{def:trace-equivalence}
  Let $\asys \in \systems{\ptypes}$ and $\asys' \in
  \systems{\ptypes^\expand}$ be two systems and $\syslab : \placeof{\ptypes}
  \rightharpoonup \vars$ be a variable labeling. Given two
  finite or infinite firing sequences
  $\rho \in \pathsof{\behof{\asys}}$ and
  $\rho' \in \pathsof{\behof{\asys'}}$,
  we say that $\rho'$ is a stuttering of $\rho$,
  and write $\rho \stutter_\syslab \rho'$,
  when there exists a finite or infinite sequence
  $A_0,A_1,\ldots \in \atoms^\infty$
  and integers $n_0, n_1, \ldots \geq 1$ such
  that: \begin{align*} \atomsof{\rho}{\syslab\circ\proj{}{1}} = & ~
    A_0,A_1,...
    \hspace*{4mm}\text{ and }\hspace*{4mm}
    \atomsof{\rho'}{\syslab^\expand\circ\proj{}{1}} =
    ~\underbrace{A_0, \ldots, A_0,}_{n_0 \text{ times }}
    \underbrace{A_1, \ldots, A_1,}_{n_1 \text{ times}} \ldots
  \end{align*}
  This relation is lifted to systems as $\asys \stutter_\syslab \asys'$,
  if and only if
  the following hold: \begin{compactenum}
  \item for each $\rho \in \pathsof{\behof{\asys}}$ there exists
    $\rho' \in \pathsof{\behof{\asys'}}$ such that $\rho
    \stutter_\syslab \rho'$,
  \item for each $\rho' \in \pathsof{\behof{\asys'}}$ there exists
    $\rho \in \pathsof{\behof{\asys}}$ such that $\rho
    \stutter_\syslab \rho'$.
  \end{compactenum}
\end{definition}

In particular $\asys \stutter \asys'$ implies that $\asys$ and $\asys'$
have sets of reachable configurations that satisfy the same atoms.

\begin{theorem}\label{thm:stutter-equiv}
  Let $\asys \in \systems{\ptypes}$ and $\asys' \in
  \systems{\ptypes^\expand}$ be systems and $\syslab :
  \placeof{\ptypes} \rightharpoonup \vars$ be a variable labeling such
  that $\asys \stutter_\syslab \asys'$. Then, for any arithmetic
  formula $\alpha$, we have $$(\asys,\syslab) \models \alpha \iff
  (\asys',\syslab^\expand) \models \alpha$$
\end{theorem}
\noindent Finally, applying \thmref{thm:stutter-equiv}
to the lemma below (proof in \appref{app:trace-equiv})
gives \thmref{thm:main}.

\begin{lemma}\label{lemma:trace-equiv}
  Let $\theta$ be a \vrtext-term over the alphabets $\valpha \isdef
  \ptypes \uplus \palpha$ and $\ealpha \isdef
  \obstransof{\ptypes} \times \obstransof{\ptypes}$, and
  $\syslab : \placeof{\ptypes} \rightarrow \vars$ be a variable
  labeling. Then $\val{\vr}{\theta} \stutter_\syslab \val{\hr}{\expand(\theta)}$.
\end{lemma}

\noindent For the proof of \thmref{thm:main}, let $\grammar'$
be the \hrtext{} grammar obtained by replacing each occurrence of a
\vrtext{} operation $\aop$ from $\grammar$ by the \hrtext{} operation
$\expand(\aop)$. ``$\Rightarrow$'' If
$\prp{\grammar}{\alpha}{\syslab}{\ptypes}$ has a positive answer, there
exists a system $\asys \in \alangof{}{\grammar}$ such that
$(\asys,\syslab) \models \alpha$. Hence, there exists a ground
\vrtext-term $\theta$ such that $\val{\vr}{\theta}=\asys$ and $\theta$
is the result of a complete derivation of $\grammar$. By the
definition of $\grammar'$, $\expand(\theta)$ is a ground
\hrtext-term resulting from a complete derivation of
$\grammar'$. By \lemref{lemma:trace-equiv},
$\val{\vr}{\theta} \stutter_\syslab \val{\hr}{\expand(\theta)}$, hence
we obtain $(\val{\hr}{\expand(\theta)},\syslab^\expand) \models \alpha$,
by \thmref{thm:stutter-equiv} and
$\prp{\grammar'}{\phi}{\syslab^\expand}{\ptypes^\expand}$ has
a positive answer. The ``$\Leftarrow$'' direction uses a symmetric
argument. \qed

\begin{textAtEnd}[category=lemma2]
\subsection*{Proof of \lemref{lemma:trace-equiv}}
\label{app:trace-equiv}

In this entire section, let $\theta$ be a fixed ground
\vrtext-term written using the set $\valpha=\ptypes\uplus\palpha$ of
vertex labels. It is easy to check that $\expand(\theta)$ is a ground
\hrtext-term written using the set
$\valpha^\expand \isdef \ptypes^\expand \uplus \palpha^\expand$
of vertex labels, where the sets $\ptypes^\expand$ of process types and
$\palpha^\expand$ of ports have been defined above.
We write $\asys$ (\resp $\asys'$) for
$\val{\vr}{\theta}$ (\resp $\val{\hr}{\expand(\theta)}$).

Take a reachable marking $\amark$ of $\behof{\asys'}$,
and any transition $t \in \obstransof{\ptype}$.
We say that a vertex $v$ of type $\epsilon_t$ is \emph{idle}
if $\amark(\gtidle_t,v) = 1$.
We say that $v$ of type $\epsilon_t$ is \emph{waiting on $t$}
if $\amark(\gtwait_t,v) = 1$.
We say that $v$ of type $\halve{\ptype}$ is \emph{waiting} on $t$
if $\amark(\halfplace{t},v) = 1$.
Finally, given a label $\elab \in \expalpha$ and an edge
$(v', \elab, v) \in \edgeof{\asys'}$, we say that
$v'$ is an $\expalpha$-predecessor of $v$.

\begin{fact}\label{fact:sameap-reachchara}
  For any reachable marking $\amark$ of $\behof{\asys'}$,
  a vertex $v$ of type $\epsilon_t$ is non-idle in $\amark$
  if and only if there exists a $\expalpha$-predecessor of $v$
  that is waiting on $t$ in $\amark$.
  Furthermore if there exists such an $\expalpha$-predecessor,
  there will be exactly one.
\end{fact}
\begin{proof}
  This proof is by induction. Since the property must hold for any reachable
  marking, we show that it holds for $\initmarkof{\behof{\asys'}}$,
  then that is it preserved by every firing of a transition.

  In $\initmarkof{\behof{\asys'}}$ all vertices of type $\epsilon_t$
  are in their initial state and thus have a token in $\gtidle_t$.
  All vertices of type $\ptype \in \ptypes$ have a token in a place
  of $\placeof{\ptypes}$, not a place from $\halfplace{(\transof{\ptypes})}$.
  Therefore in the initial configuration, non-idle processes do not exist,
  and no process is waiting.
  The property is thus easy to verify.

  For the inductive step, consider any transition that can occur in $\behof{\asys'}$,
  resulting in the firing $\amark \fire{\elab} \amark'$.
  The structural properties that we established in
  \lemref{lemma:expansion} guarantee that the following case analysis is exhaustive:
  \begin{compactitem}
  \item firing an internal transition $t$ in $v$:
    no vertex of type $\epsilon_t$ is involved, and the vertex $v$ in which
    the transition occurs is not waiting in either $\amark$ or $\amark'$
    since $\pre{t}$ and $\post{t}$ are both places of $\placeof{\ptypes}$
    not $\halfplace{(\transof{\ptypes})}$.
    Therefore the property holds in $\amark'$.
  \item firing $(v', (\gttry{t},\gtrecv), v)$:
    for this interaction to be fireable requires that
    $\amark(\pre{t},v') = 1$ ($v'$ is not waiting in $\amark$)
    and $\amark(\gtidle_t,v) = 1$ ($v$ is idle in $\amark$).
    The inductive hypothesis applied to the latter fact yields that
    all $\expalpha$-predecessors of $v$ are not waiting on $t$ in $\amark$,
    so when the firing makes $v$ non-idle and $v'$ waiting on $t$
    by $\amark'(\halfplace{t},v') = 1$ and $\amark'(\gtactive_t,v) = 1$,
    it makes $v'$ the unique $\expalpha$-predecessor of $v$ that is waiting on $t$
    in $\amark'$.
    Meanwhile $v'$ is not of type $\epsilon_t$ so is not subject to the condition,
    $v$ did not become waiting on $t$ so the possible $\expalpha$-successor
    of $v$ is not affected, and all other processes retain the same state
    so must satisfy the same property in $\amark'$ as they did in $\amark$.
    Therefore $\amark \fire{(v',(\gttry{t},\gtrecv),v)} \amark'$
    preserves the desired property.
  \item firing $(v', (\gtfwd,\gtrecv), v)$:
    in the same spirit as the previous case,
    this transition makes the token of $v'$ move from $\gtactive_t$ to $\gtwait_t$,
    so $v'$ remains non-idle and additionally becomes waiting on $t$ in $\amark'$,
    and at the same time the token of $v$ moves from $\gtidle_t$ to $\gtactive_t$,
    so $v$ becomes non-idle but remains non-waiting.

    Once again, $v$ has gained exactly one waiting $\expalpha$-predecessor,
    and all other vertices are unaffected,
    so $\amark'$ satisfies the property.
  \item firing $(v', (t',t), v)$:
    this time we must have $\amark(\gtactive_t,v') = \amark(\gtactive_{t'},v') = 1$,
    and $\amark(\gtreply_t,v') = \amark(\gtreply_{t'},v') = 1$.
    All vertices are non-idle in $\amark'$ if and only if they are non-idle
    in $\amark$,
    and they are waiting on $t$ in $\amark'$ if and only if they are waiting
    on $t$ in $\amark$.
    The desired property is thus easily preserved.
  \item firing $(v', (\gtcommit{t},\gtreset), v)$:
    this transition makes $v'$ (waiting in $\amark$) return to non-waiting in $\amark'$,
    and $v$ (non-idle in $\amark$) return to idle in $\amark'$.
    This case is exactly symmetrical to the firing of
    $(v',(\gttry{t},\gtrecv),v)$ and an analoguous reasoning shows
    that it preserves the property at hand.
  \item firing $(v', (\gtack,\gtreset), v)$:
    similarly this case is symmetrical to $(v', (\gtfwd,\gtrecv), v)$
    and accordingly can be handled by the same reasoning.
  \end{compactitem}
  Thus the property is preserved by all firings of transitions of $\behof{\asys'}$,
  and therefore holds in all reachable markings.
  \qed
\end{proof}

We define a relation $\sameap$ between markings
$\amark : \placeof{\behof{\asys}} \rightarrow \nat$
and $\amark' : \placeof{\behof{\asys'}} \rightarrow \nat$,
where $\amark \sameap \amark'$ if for every place $(q,v) \in \placeof{\behof{\asys}}$
such that $\amark(q,v) = 1$ we have one of the following properties:
\begin{enumerate}[(i)]
  \item $\amark'(q,v) = 1$, or
  \item $\exists t \in \post{q}.~ \exists v' \in \vertof{\asys'}.~ \mathsf{WaitPath}_{\amark'}(v,v') \wedge \amark'(\halfplace{t},v) = 1 \wedge \amark'(\gtactive_t,v') = 1$, or
  \item $\exists t \in \pre{q}.~ \exists v' \in \vertof{\asys'}.~ \mathsf{WaitPath}_{\amark'}(v,v') \wedge \amark'(\halfplace{t},v) = 1 \wedge \amark'(\gtreply_t,v') = 1$.
\end{enumerate}
where we define the predicate $\mathsf{WaitPath}_{\amark'}(v,v')$
as the existence of an $\expalpha$-labeled path $v,v_1,v_2,...,v_k,v'$ from $v$ to $v'$
such that $\amark'(\gtwait_t,v_i) = 1$ for every $1 \leq i \leq k$.
Furthermore we say that $\amark'$
is \emph{canonical} if only condition (i) is ever satisfied.

Intuitively $\amark'$ is canonical and $\amark \sameap \amark'$ when
$\amark'$ has tokens in the same places as $\amark$ does.
For non-canonical $\amark'$, we can tolerate transitions that are partially executed
and depending on whether the partial execution of $t$ is less or more than half
we count the token as being either in $\pre{t}$ or in $\post{t}$
(recall: an $\expalpha$-labeled path of $\asys'$ simulates a transition of $\asys$,
and the progress of this simulation is less than half if we have a path of $\gtwait_t$ until
an $\gtactive_t$, and more than half if we have a path of $\gtwait_t$ until a $\gtreply_t$).
\begin{fact}\label{fact:sameap-sameap}
  For any $\amark \sameap \amark'$, we have
  $\atomsof{\amark}{\varlab} = \atomsof{\amark'}{\varlab^\expand}$
\end{fact}
\begin{proof}
  Due to \lemref{lemma:expansion} and \factref{fact:sameap-reachchara},
  we know that when in the definition of $\sameap$ above we talk about
  an $\expalpha$-labeled path $\mathsf{WaitPath}_{\amark'}(v,v')$
  in the form of $v,v_1,...,v_k,v'$,
  it necessarily means that all $v_1,...,v_k$ are vertices of type $\epsilon_t$,
  and that they are disjoint from any other instanciation of $\mathsf{WaitPath}_{\amark'}$
  for a different place.
  In particular this implies that the conditions (i), (ii), (iii)
  are mutually exclusive, and any token in a place
  $\gtactive_t$ or $\gtreply_t$ can be injectively mapped to some vertex
  $v$ with a token in $\halfplace{t}$ from which there is a $\mathsf{WaitPath}_{\amark'}$
  as described.

  This induces a bijection between $\set{v \in \vertof{\asys} \mid \amark(q,v) = 1}$,
  and $\set{v \in \vertof{\asys'} \mid \amark'(q,v) = 1}
  \uplus \set{v \in \vertof{\asys'} \mid \amark'(\gtactive_t,v) = 1 \text{ for } t \in \post{q}}
  \uplus \set{v \in \vertof{\asys'} \mid \amark'(\gtreply_t,v) = 1 \text{ for } t \in \pre{q}}$,
  implying that they have the same cardinal.
  It happens that $\sum_{q\in\varlab^{-1}(x)} \amark(q)$ is exactly the cardinal of the first
  set, and $\sum_{q\in(\varlab^\expand)^{-1}(x)} \amark'(q)$ is the cardinal of the second.
  We conclude that $\varlab$ and $\varlab^\expand$ have the same interpretation of $x$,
  and since this applies to every $x \in \vars$ we conclude that the same atoms
  are satisfied in $\amark$ and in $\amark'$.
  \qed
\end{proof}

\begin{fact}\label{fact:sameap-epsilon}
  If $\amark'_0 \fire{\elab} \amark'_1$ in $\behof{\asys'}$
  and $\elab \in \expalpha \cup \overleftarrow{\expalpha}$,
  then for all $\amark$ marking of $\behof{\asys}$
  we have $\amark \sameap \amark'_0$ if and only if $\amark \sameap \amark'_1$.
\end{fact}
\begin{proof}
  Once again we rely on \lemref{lemma:expansion} and \factref{fact:sameap-reachchara}.
  Consider an interaction that can be fired between some $v$ and $v'$,
  and assume that it results in the transition from $\amark'_0$ to $\amark'_1$.
  Since $\sameap$ is defined pointwise on each
  place of $\behof{\asys}$ that holds a token, and since $\amark$ is unchanged,
  it suffices to look at exactly the places of $\behof{\asys'}$ that are changed
  by $\amark'_0 \fire{\elab} \amark'_1$. In all configurations we will show
  that the equivalence is preserved simply by changing which of the conditions
  (i), (ii), (iii) hold on each place.
  In all that follows we will write $t$ the transition that we consider
  (fully determined by the type $\epsilon_t$ of $v'$),
  and $v_0 \in \vertof{\asys}$ the vertex that is currently being
  considered for the equivalence.
  The cases are as follows:
  \begin{compactitem}
  \item $(v, (\gttry{t},\gtrecv), v')$:
    $\amark$ is linked to $\amark'_0$ on $(\pre{t},v)$ by condition (i).
    Furthermore $v_0 = v$.
    We show that after firing $\elab = (\gttry{t},\gtrecv)$ the markings are now
    linked by condition (ii) on the path $v,v'$.

    In $\amark'_0$, the vertex $v$ which is of type $\halve{\ptype}$
    has a token in place $\pre{t}$.
    We know this because it is the only way that its transition $\gttry{t}$ is fireable,
    and this is why condition (i) is satisfied.
    Similarly there must be a token in $(\gtidle_t,v')$ so that $\gtrecv$ is fireable.

    After firing $\elab$ we now have $\amark'_1(\halfplace{t},v) = 1$ and
    $\amark'_1(\gtactive_t,v') = 1$. This satisfies the condition for (ii)
    with $\mathsf{WaitPath}_{\amark'_1}(v,v')$ instanciated by the trivial path $v,v'$.
    Every step of this reasoning is in fact an equivalence,
    so we also get the other direction of the statement.
  \item $(v, (\gtfwd,\gtrecv), v')$:
    $\amark$ and $\amark'_0$ are linked on $(\pre{t},v_0)$ by condition (ii),
    in the form of $\mathsf{WaitPath}_{\amark'_0}(v_0,v)$ instanciated by some
    $v_0,v_1,...,v_k,v$.
    We show that $\amark$ and $\amark'_1$ are still linked by condition (ii),
    but on $v'$ instead of $v$.

    The interaction $(\gtfwd,\gtrecv)$ belongs to $\expalpha$,
    so $v_0,v_1,...,v_k,v,v'$ is still an $\expalpha$-labeled path.
    In fact after we fire $\elab$, we get
    $\amark'_1(\gtwait_t,v) = 1$ which means that
    $v_0,v_1,...,v_k,v,v'$ instanciates $\mathsf{WaitPath}_{\amark'}(v_0,v')$.
    Lastly we have $\amark'_1(\gtactive_t,v') = 1$,
    thus satisfying condition (ii).
  \item $(v, (\gtcommit{t},\gtreset), v')$:
    $\amark$ is linked to $\amark'_0$ on $(\post{t},v_0)$ by condition (iii).
    With a reasoning analoguous to the first case, including the fact that $v_0 = v$,
    we can prove that $\amark$ and $\amark'_1$ are instead linked by condition (i).
  \item $(v, (\gtack,\gtreset), v')$:
    $\amark$ is linked to $\amark'_0$ on $(\post{t},v_0)$ by condition (iii).
    With a reasoning analoguous to the second case, we can prove that
    $\amark$ and $\amark'_1$ are still linked by condition (iii),
    but using a shorter $\mathsf{WaitPath}(v_0,v)$ instead of $\mathsf{WaitPath}(v_0,v')$.
  \end{compactitem}
  \qed
\end{proof}

\begin{fact}\label{fact:sameap-stepleft}
  For any markings $\amark_1,\amark_2 : \placeof{\behof{\asys}} \rightarrow \nat$,
  transition $\elab \in \transof{\behof{\asys}}$
  such that $\amark_1 \fire{\elab} \amark_2$ in $\behof{\asys}$,
  and any $\amark'_1 : \placeof{\behof{\asys'}} \rightarrow \nat$,
  if $\amark_1 \sameap \amark'_1$ and $\amark'_1$ is canonical
  then there exist
  a marking $\amark'_2 : \placeof{\behof{\asys'}} \rightarrow \nat$,
  an $\expalpha$-sequence $\vec{p}$,
  and an $\overleftarrow{\expalpha}$-sequence $\vec{q}$,
  such that $\amark'_1 \fire{\vec{p} \cdot \elab \cdot \vec{q}} \amark'_2$
  in $\behof{\asys'}$ and $\amark_2 \sameap \amark'_2$
  with $\amark'_2$ canonical.
\end{fact}
\begin{proof}
  Firstly if $\elab$ is an internal transition,
  then the canonicity of $\amark_1 \sameap \amark_2$ and the fact that
  process types in $\halve{\ptypes}$ have the same internal transitions as
  those of $\ptypes$, give that $\elab$ must also be fireable in $\amark_2$.
  We thus choose both $\vec{p}$ and $\vec{q}$ to be the empty sequence.

  Otherwise, we assume that from $\amark_1$ to $\amark_2$,
  the interaction $\elab = (t,t')$ occurs between vertices $v$ and $v'$.
  As we hinted at in \exref{ex:halves}, we choose $\vec{p}$ to be
  $(\gttry{t},\gtrecv),(\gtfwd,\gtrecv),...,(\gtfwd,\gtrecv)$
  on an $\epsilon_t$-path starting from $v$,
  followed by
  $(\gttry{t'},\gtrecv),(\gtfwd,\gtrecv),...,(\gtfwd,\gtrecv)$,
  on an $\epsilon_{t'}$-path starting from $v'$.
  After that we execute $(t,t')$, then propagate backwards for $\vec{q}$
  as $(\gtreset,\gtack),...,(\gtreset,\gtack),(\gtreset,\gtcommit{t})$
  and $(\gtreset,\gtack),...,(\gtreset,\gtack),(\gtreset,\gtcommit{t'})$
  along the same two paths in reverse.

  Since $\amark'_1$ is canonical,
  all vertices of type $\epsilon_t$ and $\epsilon_{t'}$ in $\asys'$
  have a token in $\gtidle_t$ in $\amark'_1$.
  Since $\asys = \expof{\asys'}$, the existence of an interaction
  $(t,t')$ between $v$ and $v'$ implies that there are some
  $u,u'$ in $\asys'$ and there exist an $\epsilon_t$-path from $v$ to $u$
  and an $\epsilon_{t'}$-path from $v'$ to $u'$.
  Along such paths, the sequence of transitions described above is a firing sequence
  if initially all vertices of type $\epsilon_t$ and $\epsilon_{t'}$ are in state $\gtidle_t$
  and $\gtidle_{t'}$ respectively,
  and they will return to state $\gtidle_t$ or $\gtidle_{t'}$ once the entire
  sequence is fired.

  Thus $\amark'_1 \fire{\vec{p} \cdot \elab \cdot \vec{q}} \amark'_2$,
  $\amark'_2$ is canonical, and $\amark_2 \sameap \amark'_2$.
  \qed
\end{proof}

\begin{fact}\label{fact:sameap-seqleft}
  Every (finite or infinite) firing sequence in $\pathsof{\behof{\asys}}$
  admits a stuttering firing sequence in $\pathsof{\behof{\asys'}}$.
\end{fact}
\begin{proof}
  Let $\initmarkof{\behof{\asys}} = \amark_0 \fire{\elab_0} \amark_1 \fire{\elab_1} \amark_2 \fire{\elab_2} \cdots$
  be a (finite or infinite) firing sequence $\rho \in \pathsof{\behof{\asys}}$.
  Using the initialization that
  $\initmarkof{\behof{\asys}} \sameap \initmarkof{\behof{\asys'}}$,
  and the latter is canonical (since the initial token of $\epsilon_t$ is in $\gtidle_t$
  for all $t$),
  we can inductively apply \factref{fact:sameap-stepleft} at every step to construct
  a firing sequence
  $\initmarkof{\behof{\asys'}} =
    \amark'_0 \fire{\vec{p}_0\cdot\elab_0\vec{q}_0}
    \amark'_1 \fire{\vec{p}_1\cdot\elab_1\cdot\vec{q}_1}
    \amark'_2 \fire{\vec{p}_2\cdot\elab_2\cdot\vec{q}_2} \cdots$
  which is a path $\rho' \in \pathsof{\behof{\asys'}}$.
  In this construction all $\vec{p}_i$ and $\vec{q}_i$ are
  $\expalpha$-sequences and $\overleftarrow{\expalpha}$-sequences respectively,
  thus invariant for $\sameap$ (\factref{fact:sameap-epsilon}).
  This in turn implies that they do not change the set of atomic propositions
  that hold (\factref{fact:sameap-sameap}).

  A path $\atomsof{\rho}{\varlab} = A_0 A_1 A_2 \cdots$, where $\rho \in \pathsof{\behof{\asys}}$
  is thus transformed by this construction into a path
  $\atomsof{\rho'}{\varlab^\expand} = A_0^{1+|\vec{p}_0|} A_1^{1+|\vec{q}_0|+|\vec{p}_1|} A_2^{1+|\vec{q}_1|+|\vec{p}_2|} \cdots$
  which gives $\rho \stutter_\syslab \rho'$, for some $\rho' \in \pathsof{\behof{\asys'}}$.
  \qed
\end{proof}

\begin{fact}\label{fact:sameap-epsfinite}
  Any firing sequence of $\expalpha\cup\overleftarrow{\expalpha}$-transitions
  must be finite.
\end{fact}
\begin{proof}
  In principle, the proof is straightforward: we define a cost function on places
  and show that every $\expalpha$-labeled transition moves a token to a place that
  has a strictly smaller cost.
  Let $c : \placeof{\behof{\asys'}} \rightarrow \nat$ be defined as follows:
  \begin{align*}
    c(\gtwait_t,v) & = c(\gtidle_t,v) = c(\halfplace{t},v) = 0 & \text{for any } t,v \\
    c(\gtactive_t,v) & = 1 + c(\gtactive_t,v') & \text{if } \exists v'.~ v \xrightarrow{\expalpha} v' \\
    c(\gtactive_t,v) & = 0 & \text{otherwise} \\
    c(q,v) & = 1 + \max_{t\in\post{q}} \max_{v' : v \xrightarrow{\expalpha} v'} c(\gtactive_t,v') & \text{for any } q,v \\
    c(\gtreply_t,v) & = 1 + \max_{v' : v' \xrightarrow{\expalpha} v} c(\post{t}, v) + \max_{v': v' \xrightarrow{\expalpha} v} c(\gtreply_t,v') & \text{for any } t,v \\
  \end{align*}
  This function, apart from being obviously nonnegative, is well-defined everywhere:
  the acyclicity of the subgraph of $\asys'$ composed only of $\expalpha$-edges
  guarantees that the system of equations above has a solution.
  More precisely:
  \begin{compactitem}
  \item $c$ is obviously well-defined for places of the form $(\gtwait_t,v)$,
    $(\gtidle_t,v)$, or $(\halfplace{t},v)$;
  \item when a token is in a place $(\gtactive_t,v)$ it can only be involved in
    transitions in the direction of $\expalpha$ which is acyclic.
    Furthermore vertices $\epsilon_t$ form a forest so if $v'$ such that
    $v \xrightarrow{\expalpha} v'$ exists then it is unique,
    thus $c$ is also well-defined for places of the form $(\gtactive_t,v)$;
  \item finally $c$ over $(q,v)$ and $(\gtreply_t,v)$ is defined in terms
    of $\expalpha$-predecessors, which are not unique (thus why we need a $\max$),
    but still form an acyclic graph.
  \end{compactitem}
  Given a marking $\amark$, we define
  $$ \mathrm{fuel}(\amark) \isdef \sum_{q\in\placeof{\asys'}} \amark(q) c(q) $$
  Observe that since $\amark(q)$ can only take the value $0$ or $1$,
  the $\mathrm{fuel}$ of a marking is the sum of the cost of overy place with a token.
  Our objective is now to show that whenever
  $\amark \fire{t_0} \amark'$ for $t_0 \in \transof{\behof{\asys'}}$
  labeled by $\elab \in \expalpha\cup\overleftarrow{\expalpha}$,
  we have $\mathrm{fuel}(\amark) > \mathrm{fuel}(\amark')$.
  We do this by showing the equivalent fact that
  $\mathrm{fuel}(\amark) - \mathrm{fuel}(\amark') > 0$
  because most of it cancels out:
  partition the set of vertices into $\pre{t_0}$ (loses a token when we fire $t_0$),
  $\post{t_0}$ (receives a token when we fire $t_0$),
  and the rest (unchanged), and we are left with simply
  $$ \mathrm{fuel}(\amark) - \mathrm{fuel}(\amark')
    = \sum_{q \in \pre{t}} c(q) - \sum_{q \in \post{t}} c(q) $$
  Now for the case analysis on $\elab$:
  \begin{compactitem}
  \item $(v, (\gttry{t},\gtrecv), v')$:
    $\mathrm{fuel}(\amark) - \mathrm{fuel}(\amark')
    = c(\pre{t},v) + c(\gtidle_t,v') - c(\halfplace{t},v) - c(\gtactive_t,v')
    = c(\pre{t},v) - c(\gtactive_t,v')$.
    Since $v \xrightarrow{\expalpha} v'$, $(\gtactive_t,v')$ is already
    accounted for in $c(\pre{t}, v) = 1 + \max_{t\in \post{(\pre{t})}}
    \max_{v':v\xrightarrow{\expalpha} v'} c(\gtactive_t,v')$,
    and thus $c(\pre{t},v) > c(\gtactive_t,v')$.
    We easily conclude that $\mathrm{fuel}(\amark) - \mathrm{fuel}(\amark') > 0$.
  \item $(v, (\gtfwd,\gtrecv), v')$:
    $\mathrm{fuel}(\amark) - \mathrm{fuel}(\amark')
    = c(\gtactive_t,v) + c(\gtidle_t,v') - c(\gtidle_t,v) - c(\gtactive_t,v')
    = c(\gtactive_t,v) - c(\gtactive_t,v')
    = (1 + c(\gtactive_t,v')) - c(\gtactive_t,v')
    = 1$ strictly positive as expected.
  \item $(v, (\gtcommit{t},\gtreset), v')$:
    $\mathrm{fuel}(\amark) - \mathrm{fuel}(\amark')
    = c(\halfplace{t},v) + c(\gtreply_t,v') - c(\post{t},v) - c(\gtidle_t,v')
    = c(\gtreply_t,v') - c(\post{t},v)$.
    Since the quantity $c(\gtreply_t,v)$ occurs in
    $\max_{v:v\xrightarrow{\expalpha}v'} c(\gtreply_t,v)$
    which is part of the sum that constructs $c(\gtreply_t,v')$,
    the inequality is obvious.
  \item $(v, (\gtack,\gtreset), v')$:
    The same reasoning as the previous case applies, since this time
    $\mathrm{fuel}(\amark) - \mathrm{fuel}(\amark')
    = c(\gtreply_t,v) - c(\gtreply_t,v')$ and once again
    $c(\gtreply_t,v')$ occurs in
    $\max_{v:v\xrightarrow{\expalpha}v'} c(\gtreply_t,v)$
    leading to $c(\gtreply_t,v) - c(\gtreply_t,v') > 0$.
  \end{compactitem}
  In all cases we obtain $\mathrm{fuel}(\amark) > \mathrm{fuel}(\amark')$.
  The function $\mathrm{fuel}$, as we have already justified,
  is a well-defined finite nonnegative quantity.
  It follows that an infinite firing sequence of
  $\expalpha\cup\overleftarrow{\expalpha}$-labeled edges cannot exist.
  \qed
\end{proof}

\begin{fact}\label{fact:sameap-stepright}
  For any markings $\amark'_1,\amark'_2 : \placeof{\behof{\asys'}} \rightarrow \nat$,
  transition $\elab \in \transof{\behof{\asys'}}$
  such that $\amark'_1 \fire{\elab} \amark'_2$ in $\behof{\asys}$,
  and any $\amark_1 : \placeof{\behof{\asys}} \rightarrow \nat$,
  if $\amark_1 \sameap \amark'_1$ and $\delta\not\in\expalpha\cup\overleftarrow{\expalpha}$
  then there exists a marking $\amark_2 : \placeof{\behof{\asys}} \rightarrow \nat$
  such that $\amark_1 \fire{\elab} \amark_2$ in $\behof{\asys}$
  and $\amark_2 \sameap \amark'_2$.
\end{fact}
\begin{proof}
  If $\elab = t$ is an internal transition,
  it means that for some $(q,v) = \pre{t}$ we have $\amark'_1(q,v) = 1$.
  Thee definition of $\sameap$ this implies $\amark_1(q,v) = 1$ and thus $\elab$
  is also fireable in $\behof{\asys}$ and obviously leads to $\amark_2 \sameap \amark'_2$.

  Otherwise $\elab = (t,t')$ is an observable transition.
  It must occur between some $e$ and $e'$ of respective types $\epsilon_t$
  and $\epsilon_{t'}$, and in $\amark_2$ they must have a token in
  $\gtactive_t$ and $\gtactive_{t'}$ respectively.
  By the characterization from \factref{fact:sameap-reachchara},
  there must furthermore exist some vertices $v, v'$ such that
  there is an $\epsilon_t$-path from $v$ to $e$
  and an $\epsilon_{t'}$-path from $v'$ to $e'$,
  and all vertices along those paths have a token
  in $\gtwait_t$ or $\gtwait_{t'}$ respectively.
  Finally, still from \factref{fact:sameap-reachchara},
  we obtain that in $\amark'_1$, $v$ and $v'$ have a token in
  $\halfplace{t}$ and $\halfplace{t'}$ respectively.

  Knowing that $\amark_1 \sameap \amark'_1$,
  this means that $\amark_1(\pre{t},v) = 1$ and $\amark_1(\pre{t'},v) = 1$,
  and thus $(t,t')$ is fireable in $\amark'_1$.
  Firing it yields $\amark'_2$ with
  $\amark_2(\post{t},v) = 1$ and $\amark_2(\post{t'},v) = 1$.
  On the other hand $\amark'_2$ only differs from $\amark'_1$ in that
  instead of having
  $\amark'_1(\gtactive_t,e) = 1$ and $\amark'_1(\gtactive_{t'},e') = 1$
  we now have
  $\amark'_2(\gtreply_t,e) = 1$ and $\amark'_2(\gtreply_{t'},e') = 1$.
  In the definition of $\sameap$, we thus easily check that
  $\amark_2 \sameap \amark'_2$.
  \qed
\end{proof}

\begin{fact}\label{fact:sameap-seqright}
  Every (finite or infinite) firing sequence in $\pathsof{\behof{\asys'}}$
  is a stuttering of a firing sequence in $\pathsof{\behof{\asys}}$.
\end{fact}
\begin{proof}
  Let $\initmarkof{\behof{\asys'}} = \amark'_0 \fire{t_1} \amark'_1 \fire{t_2} \cdots$
  be a firing sequence $\rho' \in \pathsof{\behof{\asys'}}$.
  Let $i_0 = 0$ then $i_1,i_2,...$ be an enumeration of all the indices $i$ for which
  $t_i \not\in \expalpha\cup\overleftarrow{\expalpha}$.
  Starting from $\initmarkof{\behof{\asys'}} = \amark_0$,
  we apply \factref{fact:sameap-stepright} for every $i_j$ in increasing order,
  to get a firing sequence $\rho$ defined by
  $\amark_0 \fire{t_{i_1}} \amark_1 \fire{t_{i_2}} \cdots$
  and such that $\amark_j \sameap \amark'_{i_j}$ for every $j \geq 0$.
  Applying \factref{fact:sameap-epsilon} in-between gives that
  $\amark_j \sameap \amark'_k$ for every $i_j \leq k < i_{j+1}$,
  and \factref{fact:sameap-epsfinite} guarantees that there are finitely
  many stuttering steps.
  We can finally conclude
  $\atomsof{\rho}{\varlab} \stutter \atomsof{\rho'}{\varlab^\expand}$.
  \qed
\end{proof}

Finally we conclude from \factref{fact:sameap-seqleft} and \factref{fact:sameap-seqright}
that $\behof{\asys} \stutter_\syslab \behof{\asys'}$:
$\behof{\asys}$ and $\behof{\asys'}$ have the same set of paths,
up to finite stuttering on the side of $\asys'$.
\end{textAtEnd}

%% file: figure-halfprocs.tex
\begin{figure}[t!]
  \vspace*{-\baselineskip}
  \begin{center}
    \begin{minipage}{3.5cm}
      \begin{center}
      \scalebox{0.7}{
        \begin{tikzpicture}
          \tikzset{
            arr/.style={->,thick,line width=1pt}
          }
          \node[petri-p,label=-90:{$\mathsf{lo}$}] (lo) at (0,0) {};
          \node[petri-tok] at (lo) {};
          \node[petri-p,label=90:{$\mathsf{hi}$}] (hi) at ($(lo) + (0,2)$) {};
          \node[petri-thor,fill=black,label=0:{$\mathsf{up}$}] (up) at ($(lo) + (1,1)$) {};
          \node[petri-thor,fill=black,label=180:{$\mathsf{dn}$}] (dn) at ($(lo) + (-1,1)$) {};
          \draw[arr] (lo) edge[out=0,in=-90] (up);
          \draw[arr] (up) edge[out=90,in=0] (hi);
          \draw[arr] (hi) edge[out=180,in=90] (dn);
          \draw[arr] (dn) edge[out=-90,in=180] (lo);
        \end{tikzpicture}
      }
      \end{center}
       \vspace*{-1\baselineskip}
      \centerline{(a)}
    \end{minipage}
    \begin{minipage}{4.5cm}
      \begin{center}
      \scalebox{0.7}{
        \begin{tikzpicture}
          \tikzset{
            arr/.style={->,thick,line width=1pt}
          }
          \node[petri-p,label=-90:{$\mathsf{lo}$}] (lo) at (0,0) {};
          \node[petri-tok] at (lo) {};
          \node[petri-p,label=90:{$\mathsf{hi}$}] (hi) at ($(lo) + (0,2)$) {};
          \node[petri-p,label=0:{$\halfplace{\mathsf{up}}$}] (up) at ($(lo) + (1,1)$) {};
          \node[petri-p,label=180:{$\halfplace{\mathsf{dn}}$}] (dn) at ($(lo) + (-1,1)$) {};
          \node[petri-tneg,fill=black,label=-60:{$\gttry{\mathsf{up}}$}] (up1) at ($(lo)!0.5!(up) + (0.3,-0.3)$) {};
          \node[petri-tpos,fill=black,label=60:{$\gtcommit{\mathsf{up}}$}] (up2) at ($(up)!0.5!(hi) + (0.3,0.3)$) {};
          \node[petri-tneg,fill=black,label=60:{$\gttry{\mathsf{dn}}$}] (dn1) at ($(dn)!0.5!(hi) + (-0.3,0.3)$) {};
          \node[petri-tpos,fill=black,label=-60:{$\gtcommit{\mathsf{dn}}$}] (dn2) at ($(lo)!0.5!(dn) + (-0.3,-0.3)$) {};
          \draw[arr] (lo) edge[bend right=20] (up1);
          \draw[arr] (up1) edge[bend right=20] (up);
          \draw[arr] (up) edge[bend right=20] (up2);
          \draw[arr] (up2) edge[bend right=20] (hi);
          \draw[arr] (hi) edge[bend right=20] (dn1);
          \draw[arr] (dn1) edge[bend right=20] (dn);
          \draw[arr] (dn) edge[bend right=20] (dn2);
          \draw[arr] (dn2) edge[bend right=20] (lo);
        \end{tikzpicture}
      }
      \end{center}
       \vspace*{-1\baselineskip}
      \centerline{(b)}
    \end{minipage}
    \begin{minipage}{4cm}
       \begin{center}
         \scalebox{0.7}{
        \begin{tikzpicture}
          \tikzset{
            arr/.style={->,thick,line width=1pt}
          }
          \node[petri-p,label=-90:{$\mathsf{idle}$}] (idle) at (0,0) {};
          \node[petri-tok] at (idle) {};
          \node[petri-p,label=90:{$\mathsf{wait}$}] (wait) at ($(idle) + (0,2)$) {};
          \node[petri-p,label=0:{$\mathsf{active}$}] (active) at ($(idle) + (1,1)$) {};
          \node[petri-p,label=180:{$\mathsf{reply}$}] (reply) at ($(idle) + (-1,1)$) {};
          \node[petri-tneg,fill=black,label=-60:{$\gtrecv$}] (recv) at ($(idle)!0.5!(active) + (0.3,-0.3)$) {};
          \node[petri-tpos,fill=black,label=60:{$\gtfwd$}] (fwd) at ($(wait)!0.5!(active) + (0.3,0.3)$) {};
          \node[petri-tneg,fill=black,label=60:{$\gtack$}] (ack) at ($(wait)!0.5!(reply) + (-0.3,0.3)$) {};
          \node[petri-tpos,fill=black,label=-60:{$\gtreset$}] (reset) at ($(idle)!0.5!(reply) + (-0.3,-0.3)$) {};
          \node[petri-tver,fill=black,label=-90:{$t$}] (exec) at ($(reply)!0.5!(active)$) {};
          \draw[arr] (idle) edge[bend right=20] (recv);
          \draw[arr] (recv) edge[bend right=20] (active);
          \draw[arr] (active) edge[bend right=20] (fwd);
          \draw[arr] (fwd) edge[bend right=20] (wait);
          \draw[arr] (wait) edge[bend right=20] (ack);
          \draw[arr] (ack) edge[bend right=20] (reply);
          \draw[arr] (reply) edge[bend right=20] (reset);
          \draw[arr] (reset) edge[bend right=20] (idle);
          \draw[arr] (active) edge (exec);
          \draw[arr] (exec) edge (reply);
        \end{tikzpicture}
         }
       \end{center}
       \vspace*{-1\baselineskip}
       \centerline{(c)}
    \end{minipage}
  \end{center}
  \vspace*{-1.7\baselineskip}
  \caption{A process type $\ptype$ (a), the corresponding process type
    $\halve{\ptype}$ (b), and $\epsilon_t$ (c).
  }
  \label{fig:halfprocs}
  \vspace*{-\baselineskip}
\end{figure}
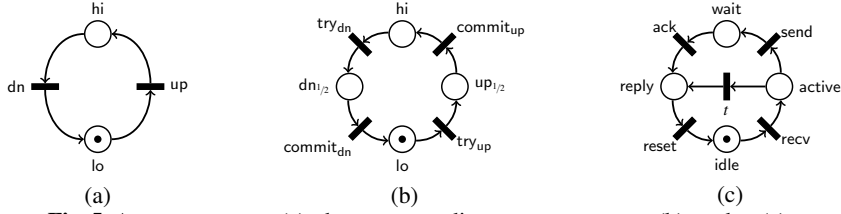

%% file: figure-router.tex
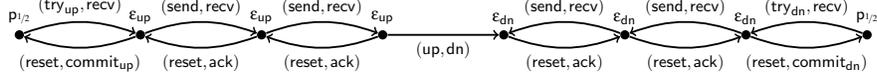
\begin{figure}[t!]
  \vspace*{-\baselineskip}
  \scalebox{0.7}{
  \hspace{1.5em}\begin{tikzpicture}
    \tikzset{
      arr/.style={->,thick,line width=1pt},
      proc/.style={circle,fill=black,inner sep=0pt,minimum width=5pt},
    }
    \pgfmathsetmacro\horiz{2.3}
    \pgfmathtruncatemacro\nb{2}
    \foreach \i in {0,...,\nb} {
      \node[proc,label=90:{$\epsilon_\mathsf{up}$}] (el\i) at (-\horiz*\i - 0.5*\horiz,0) {};
      \node[proc,label=90:{$\epsilon_\mathsf{dn}$}] (er\i) at (\horiz*\i + 0.5*\horiz,0) {};

    }
    \node[proc,label=90:{$\halve{\ptype}$}] (pl) at ($(el\nb) - (\horiz,0)$) {};
    \node[proc,label=90:{$\halve{\ptype}$}] (pr) at ($(er\nb) + (\horiz,0)$) {};
    \foreach \i in {1,...,\nb} {
      \pgfmathtruncatemacro\prev{\i - 1}
      \draw[arr] (el\i) to[bend left=20] node[above]{$(\gtfwd,\gtrecv)$} (el\prev);
      \draw[arr] (el\prev) to[bend left=20] node[below]{$(\gtreset,\gtack)$} (el\i);
      \draw[arr] (er\i) to[bend right=20] node[above]{$(\gtfwd,\gtrecv)$} (er\prev);
      \draw[arr] (er\prev) to[bend right=20] node[below]{$(\gtreset,\gtack)$} (er\i);
    }
    \draw[arr] (pl) to[bend left=20] node[above]{$(\gttry{\mathsf{up}},\gtrecv)$} (el\nb);
    \draw[arr] (el\nb) to[bend left=20] node[below]{$(\gtreset,\gtcommit{\mathsf{up}})$} (pl);
    \draw[arr] (pr) to[bend right=20] node[above]{$(\gttry{\mathsf{dn}},\gtrecv)$} (er\nb);
    \draw[arr] (er\nb) to[bend right=20] node[below]{$(\gtreset,\gtcommit{\mathsf{dn}})$} (pr);
    \draw[arr] (el0) -- node[below]{$(\mathsf{up},\mathsf{dn})$} (er0);
  \end{tikzpicture}
  }

  \vspace*{-\baselineskip}
  \caption{An $\epsilon_{\mathsf{up}}$-path meeting an $\epsilon_{\mathsf{dn}}$-path
    to simulate the interaction $(\mathsf{up},\mathsf{dn})$ between two
    vertices of type $\ptype$ from Figure \ref{fig:halfprocs} (a).}
  \label{fig:routing-gadget}
  \vspace*{-\baselineskip}
\end{figure}

%% file: figure-h.tex
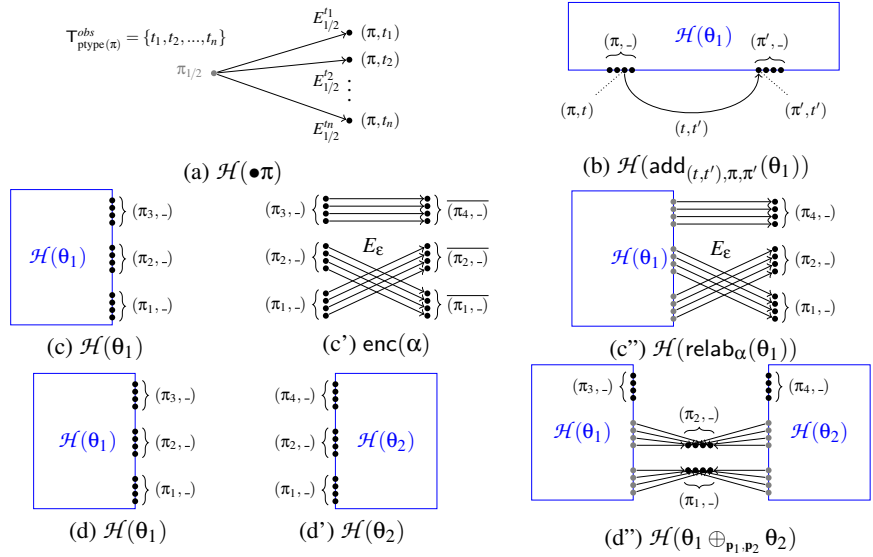
\begin{figure}[t!]
  \vspace*{-\baselineskip}
  \begin{minipage}{0.5\textwidth}
    \begin{center}
    \scalebox{0.9}{
    \begin{tikzpicture}
      \tikzset{
        proc/.style={circle,fill=black,inner sep=0pt,minimum width=2.5pt},
        arr/.style={->,line width=0.3pt},
      }
      \node[proc,fill=gray,label=-180:{\scriptsize\color{gray}$\halve{\plab}$}] (p) at (0,0) {};
      \node[proc,label=0:{\scriptsize$(\plab,t_1)$}] (p1) at (2,0.6) {};
      \node[proc,label=0:{\scriptsize$(\plab,t_2)$}] (p2) at (2,0.2) {};
      \node at (2,-0.15) {$\vdots$};
      \node[proc,label=0:{\scriptsize$(\plab,t_n)$}] (pn) at (2,-0.7) {};

      \node at (-1,0.5) {\scriptsize$\obstransof{\ptypeof{}(\plab)} = \set{t_1,t_2,...,t_n}$};

      \draw[arr] (p) to node[above,pos=0.85]{\scriptsize$\halve{E}^{t_1}$} (p1);
      \draw[arr] (p) to node[below,pos=0.85]{\scriptsize$\halve{E}^{t_2}$} (p2);
      \draw[arr] (p) to node[below,pos=0.85]{\scriptsize$\halve{E}^{t_n}$} (pn);
    \end{tikzpicture}
    }
    \centerline{(a) $\expand(\svertex{\plab})$}
    \end{center}
  \end{minipage}
  \begin{minipage}{0.5\textwidth}
    \begin{center}
    \scalebox{0.9}{
    \begin{tikzpicture}
      \tikzset{
        proc/.style={circle,fill=black,inner sep=0pt,minimum width=2.5pt},
        arr/.style={->,line width=0.3pt},
        brace/.style={decorate,decoration={brace,mirror,raise=3pt}},
      }
      \pgfmathsetmacro\mini{0.11}
      \draw[blue] (0,0) rectangle (4,1);
      \node at (2,0.5) {\color{blue}$\expand(\theta_1)$};

      \foreach \pid in {1,...,4} {
        \node[proc] (t\pid-l) at (0.5+\pid*\mini,0) {};
        \node[proc] (t\pid-r) at (2.7+\pid*\mini,0) {};
      }
      \draw[decorate,decoration={brace,mirror,raise=5pt}] ($(t4-l) + (0.5*\mini,0)$) to node[above=5pt]{\scriptsize$(\plab,\_)$} ($(t1-l) + (-0.5*\mini,0)$);
      \draw[decorate,decoration={brace,mirror,raise=3pt}] ($(t4-r) + (0.5*\mini,0)$) to node[above=5pt]{\scriptsize$(\plab',\_)$} ($(t1-r) + (-0.5*\mini,0)$);

      \node(llab) at ($(t3-l) + (-0.7,-0.6)$) {\scriptsize$(\plab,t)$};
      \node(rlab) at ($(t1-r) + (0.7,-0.6)$) {\scriptsize$(\plab',t')$};
      \draw[densely dotted] (llab) to (t3-l);
      \draw[densely dotted] (rlab) to (t1-r);
      \draw[arr] (t3-l) to[out=-90,in=-90] node[below]{\scriptsize$(t,t')$} (t1-r);
    \end{tikzpicture}
  }
    \centerline{(b) $\expand(\addedge{(t,t')}{\plab}{\plab'}{}(\theta_1))$}
    \end{center}
  \end{minipage}

  \begin{minipage}{0.2\textwidth}
    \begin{center}
    \scalebox{0.9}{
    \begin{tikzpicture}
      \tikzset{
        proc/.style={circle,fill=black,inner sep=0pt,minimum width=2.5pt},
        arr/.style={->,line width=0.3pt},
        brace/.style={decorate,decoration={brace,mirror,raise=3pt}},
      }
      \pgfmathsetmacro\step{0.7}
      \pgfmathsetmacro\mini{0.11}
      \draw[blue] (0,0) rectangle (-1.5,2);
      \node at (-0.8,1) {\color{blue}$\expand(\theta_1)$};

      \foreach \pid/\pos/\label in {1/0/1,2/1/2,3/2/3} {
        \node[proc] (p\pid-t1l) at (0,\pos*\step + 1*\mini) {};
        \node[proc] (p\pid-t2l) at (0,\pos*\step + 2*\mini) {};
        \node[proc] (p\pid-t3l) at (0,\pos*\step + 3*\mini) {};
        \node[proc] (p\pid-t4l) at (0,\pos*\step + 4*\mini) {};
        \draw[brace] ($(p\pid-t1l) + (0,-0.5*\mini)$) to node[right=5pt]{\scriptsize$(\plab_\label,\_)$} ($(p\pid-t4l) + (0,0.5*\mini)$);
      }
    \end{tikzpicture}
  }
    \centerline{(c) $\expand(\theta_1)$}
    \end{center}
  \end{minipage}
  \begin{minipage}{0.4\textwidth}
    \begin{center}
    \scalebox{0.9}{
    \begin{tikzpicture}
      \tikzset{
        proc/.style={circle,fill=black,inner sep=0pt,minimum width=2.5pt},
        arr/.style={->,line width=0.3pt},
        brace/.style={decorate,decoration={brace,mirror,raise=3pt}},
      }
      \pgfmathsetmacro\step{0.7}
      \pgfmathsetmacro\mini{0.11}
      \foreach \pid/\pos/\label in {1/0/1,2/1/2,3/2/3} {
        \node[proc] (p\pid-t1l) at (0,\pos*\step + 1*\mini) {};
        \node[proc] (p\pid-t2l) at (0,\pos*\step + 2*\mini) {};
        \node[proc] (p\pid-t3l) at (0,\pos*\step + 3*\mini) {};
        \node[proc] (p\pid-t4l) at (0,\pos*\step + 4*\mini) {};
        \draw[brace] ($(p\pid-t4l) + (0,0.5*\mini)$) to node[left=5pt]{\scriptsize$(\plab_\label,\_)$} ($(p\pid-t1l) + (0,-0.5*\mini)$);
      }

      \foreach \pid/\pos/\prev/\label in {1/0/2/1,2/1/1/2,3/2/3/4} {
        \node[proc] (p\pid-t1r) at (1.5,\pos*\step + 1*\mini) {};
        \node[proc] (p\pid-t2r) at (1.5,\pos*\step + 2*\mini) {};
        \node[proc] (p\pid-t3r) at (1.5,\pos*\step + 3*\mini) {};
        \node[proc] (p\pid-t4r) at (1.5,\pos*\step + 4*\mini) {};
        \draw[brace] ($(p\pid-t1r) + (0,-0.5*\mini)$) to node[right=5pt]{\scriptsize$\overline{(\plab_\label,\_)}$} ($(p\pid-t4r) + (0,0.5*\mini)$);

        \draw[arr] (p\prev-t1l) to (p\pid-t1r);
        \draw[arr] (p\prev-t2l) to (p\pid-t2r);
        \draw[arr] (p\prev-t3l) to (p\pid-t3r);
        \draw[arr] (p\prev-t4l) to (p\pid-t4r);
      }
      \node[below right of=p3-t4l] {$E_\epsilon$};
    \end{tikzpicture}
  }
    \centerline{(c') $\mathsf{enc}(\alpha)$}
    \end{center}
  \end{minipage}
  \begin{minipage}{0.3\textwidth}
    \begin{center}
    \scalebox{0.9}{
    \begin{tikzpicture}
      \tikzset{
        proc/.style={circle,fill=black,inner sep=0pt,minimum width=2.5pt},
        arr/.style={->,line width=0.3pt},
        brace/.style={decorate,decoration={brace,mirror,raise=3pt}},
      }
      \pgfmathsetmacro\step{0.7}
      \pgfmathsetmacro\mini{0.11}
      \draw[blue] (0,0) rectangle (-1.5,2);
      \node at (-0.5,1) {\color{blue}$\expand(\theta_1)$};

      \foreach \pid/\pos/\label in {1/0/1,2/1/2,3/2/3} {
        \node[proc,fill=gray] (p\pid-t1l) at (0,\pos*\step + 1*\mini) {};
        \node[proc,fill=gray] (p\pid-t2l) at (0,\pos*\step + 2*\mini) {};
        \node[proc,fill=gray] (p\pid-t3l) at (0,\pos*\step + 3*\mini) {};
        \node[proc,fill=gray] (p\pid-t4l) at (0,\pos*\step + 4*\mini) {};
      }

      \foreach \pid/\pos/\prev/\label in {1/0/2/1,2/1/1/2,3/2/3/4} {
        \node[proc] (p\pid-t1r) at (1.5,\pos*\step + 1*\mini) {};
        \node[proc] (p\pid-t2r) at (1.5,\pos*\step + 2*\mini) {};
        \node[proc] (p\pid-t3r) at (1.5,\pos*\step + 3*\mini) {};
        \node[proc] (p\pid-t4r) at (1.5,\pos*\step + 4*\mini) {};
        \draw[brace] ($(p\pid-t1r) + (0,-0.5*\mini)$) to node[right=5pt]{\scriptsize$(\plab_\label,\_)$} ($(p\pid-t4r) + (0,0.5*\mini)$);

        \draw[arr] (p\prev-t1l) to (p\pid-t1r);
        \draw[arr] (p\prev-t2l) to (p\pid-t2r);
        \draw[arr] (p\prev-t3l) to (p\pid-t3r);
        \draw[arr] (p\prev-t4l) to (p\pid-t4r);
      }
      \node[below right of=p3-t4l] {$E_\epsilon$};
    \end{tikzpicture}
  }
    \centerline{(c'') $\expand(\relab{\alpha}{}(\theta_1))$}
    \end{center}
  \end{minipage}

  \begin{minipage}{0.25\textwidth}
    \begin{center}
    \scalebox{0.9}{
    \begin{tikzpicture}
      \tikzset{
        proc/.style={circle,fill=black,inner sep=0pt,minimum width=2.5pt},
        arr/.style={->,line width=0.3pt},
        brace/.style={decorate,decoration={brace,mirror,raise=3pt}},
      }
      \pgfmathsetmacro\step{0.7}
      \pgfmathsetmacro\mini{0.11}
      \draw[blue] (0,0) rectangle (-1.5,2);
      \node at (-0.7,1) {\color{blue}$\expand(\theta_1)$};

      \foreach \pid in {1,...,4} {
        \node[proc] (p1t\pid-l) at (0,0*\step + \pid*\mini) {};
        \node[proc] (p2t\pid-l) at (0,1*\step + \pid*\mini) {};
        \node[proc] (p3t\pid-l) at (0,2*\step + \pid*\mini) {};
      }
      \draw[brace] ($(p1t1-l) + (0,-0.5*\mini)$) to node[right=5pt]{\scriptsize$(\plab_1,\_)$} ($(p1t4-l) + (0,0.5*\mini)$);
      \draw[brace] ($(p2t1-l) + (0,-0.5*\mini)$) to node[right=5pt]{\scriptsize$(\plab_2,\_)$} ($(p2t4-l) + (0,0.5*\mini)$);
      \draw[brace] ($(p3t1-l) + (0,-0.5*\mini)$) to node[right=5pt]{\scriptsize$(\plab_3,\_)$} ($(p3t4-l) + (0,0.5*\mini)$);
    \end{tikzpicture}
  }
    \centerline{(d) $\expand(\theta_1)$}
    \end{center}
  \end{minipage}
  \begin{minipage}{0.25\textwidth}
    \begin{center}
    \scalebox{0.9}{
    \begin{tikzpicture}
      \tikzset{
        proc/.style={circle,fill=black,inner sep=0pt,minimum width=2.5pt},
        arr/.style={->,line width=0.3pt},
        brace/.style={decorate,decoration={brace,mirror,raise=3pt}},
      }
      \pgfmathsetmacro\step{0.7}
      \pgfmathsetmacro\mini{0.11}
      \draw[blue] (2,0) rectangle (3.5,2);
      \node at (2.75,1) {\color{blue}$\expand(\theta_2)$};

      \foreach \pid in {1,...,4} {
        \node[proc] (p1t\pid-r) at (2,0*\step+5*\mini + -\pid*\mini) {};
        \node[proc] (p2t\pid-r) at (2,1*\step+5*\mini + -\pid*\mini) {};
        \node[proc] (p3t\pid-r) at (2,2*\step+5*\mini + -\pid*\mini) {};
      }
      \draw[brace] ($(p1t1-r) + (0,0.5*\mini)$) to node[left=5pt]{\scriptsize$(\plab_1,\_)$} ($(p1t4-r) + (0,-0.5*\mini)$);
      \draw[brace] ($(p2t1-r) + (0,0.5*\mini)$) to node[left=5pt]{\scriptsize$(\plab_2,\_)$} ($(p2t4-r) + (0,-0.5*\mini)$);
      \draw[brace] ($(p3t1-r) + (0,0.5*\mini)$) to node[left=5pt]{\scriptsize$(\plab_4,\_)$} ($(p3t4-r) + (0,-0.5*\mini)$);
    \end{tikzpicture}
  }
    \centerline{(d') $\expand(\theta_2)$}
    \end{center}
  \end{minipage}
  \begin{minipage}{0.5\textwidth}
    \begin{center}
    \scalebox{0.9}{
    \begin{tikzpicture}
      \tikzset{
        proc/.style={circle,fill=black,inner sep=0pt,minimum width=2.5pt},
        arr/.style={->,line width=0.3pt},
        brace/.style={decorate,decoration={brace,mirror,raise=3pt}},
      }
      \pgfmathsetmacro\step{0.7}
      \pgfmathsetmacro\mini{0.11}
      \draw[blue] (0,0) rectangle (-1.5,2);
      \node at (-0.75,1) {\color{blue}$\expand(\theta_1)$};
      \draw[blue] (2,0) rectangle (3.5,2);
      \node at (2.75,1) {\color{blue}$\expand(\theta_2)$};

      \foreach \pid in {1,...,4} {
        \node[proc,fill=gray] (p1t\pid-l) at (0,0*\step + \pid*\mini) {};
        \node[proc,fill=gray] (p2t\pid-l) at (0,1*\step + \pid*\mini) {};
        \node[proc] (p3t\pid-l) at (0,2*\step + \pid*\mini) {};
      }
      \draw[brace] ($(p3t4-l) + (0,0.5*\mini)$) to node[left=5pt]{\scriptsize$(\plab_3,\_)$} ($(p3t1-l) + (0,-0.5*\mini)$);
      \foreach \pid in {1,...,4} {
        \node[proc,fill=gray] (p1t\pid-r) at (2,0*\step+5*\mini + -\pid*\mini) {};
        \node[proc,fill=gray] (p2t\pid-r) at (2,1*\step+5*\mini + -\pid*\mini) {};
        \node[proc] (p3t\pid-r) at (2,2*\step+5*\mini + -\pid*\mini) {};
      }
      \draw[brace] ($(p3t4-r) + (0,-0.5*\mini)$) to node[right=5pt]{\scriptsize$(\plab_4,\_)$} ($(p3t1-r) + (0,0.5*\mini)$);

      \foreach \pid in {1,...,4} {
        \node[proc] (p1t\pid-m) at (0.7+5*\mini-\pid*\mini,0*\step+4*\mini) {};
        \draw[arr] (p1t\pid-l) to (p1t\pid-m);
        \draw[arr] (p1t\pid-r) to (p1t\pid-m);
        \node[proc] (p2t\pid-m) at (0.7+\pid*\mini,1*\step+\mini) {};
        \draw[arr] (p2t\pid-l) to (p2t\pid-m);
        \draw[arr] (p2t\pid-r) to (p2t\pid-m);
      }
      \draw[brace,decoration={raise=5pt}] ($(p1t4-m) + (-0.5*\mini,0)$) to node[below=6pt]{\scriptsize$(\plab_1,\_)$} ($(p1t1-m) + (0.5*\mini,0)$);
      \draw[brace,decoration={raise=5pt}] ($(p2t4-m) + (0.5*\mini,0)$) to node[above=6pt]{\scriptsize$(\plab_2,\_)$} ($(p2t1-m) + (-0.5*\mini,0)$);
    \end{tikzpicture}
  }
    \centerline{(d'') $\expand(\theta_1 \vrpop{\plabs_1}{\plabs_2} \theta_2)$}
    \end{center}
  \end{minipage}
  \vspace*{-.5\baselineskip}
  \caption{Construction of $\expand$ for each \vrtext{} symbol.
    Relabeling uses $\alpha = [\plab_1 \mapsto \plab_2, \plab_2 \mapsto \plab_1, \plab_3 \mapsto \plab_4]$,
    and composition uses $\plabs_1 = \set{\plab_1,\plab_2,\plab_3}$ and $\plabs_2 = \set{\plab_1,\plab_2,\plab_4}$.
    Ports in gray are hidden during a relabeling.
  }
  \label{fig:h}
  \vspace*{-\baselineskip}
\end{figure}

%% file: conclusions.tex
\section{Conclusions and Future Work}

We consider the parametric reachability problem for families of
networks of finite-state processes specified using \vrtext{} graph
grammars. Our approach is to reduce the problem to \hrtext{} grammars,
for which several verification techniques have been developed
lately. The reduction is based on a construction of an \hrtext{}
grammar that expands to the language of a given \vrtext{} grammar, by
elimination of epsilon-edges. We prove that this construction
preserves any reachability property expressed using first-order
arithmetic over variables that count the number of processes in each
state.

Future work involves an implementation of the proposed reduction,
based on existing tools for parametric verification of networks
specified using \hrtext{} grammars~\cite{arXiv2025}. In particular, we
plan to experiment with Azure-like network topologies and suitable
routing (existence of an available route) properties. We also plan to extend the decidability result for
pebble-passing systems with \hrtext{} grammars~\cite{arXiv2025} to
\vrtext{} grammars.